\documentclass{article}

\usepackage{amsthm}
\usepackage{amssymb}
\usepackage{amsmath}

\usepackage{graphicx}
\usepackage{url}
\usepackage[hidelinks]{hyperref}
\usepackage[dvipsnames*,svgnames]{xcolor}
\hypersetup{colorlinks=true,allcolors=DarkBlue,linkcolor=black}

\usepackage{enumerate}

\newcommand{\FTCpm}{\ensuremath{\mathsf{AxFTCpm}}}
\newcommand{\FTCcm}{\ensuremath{\mathsf{AxFTCcm}}}

\newcommand{\axname}[1]{\textup{\ensuremath{\textrm{#1}}}}

\newcommand{\CLSCL}{\axname{C$\ell$SCL}}
\newcommand{\CLSCLtwo}{\axname{C$\ell$SCL$_2$}}

\newcommand{\CLe}{\axname{EqCL}}

\newcommand{\notmodels}{\mathbin{{\,!\,}{\models}}}
\newcommand{\unmodels}{\mathbin{\hspace{.21mm}{\uparrow}\hspace{.15mm}{\models}}}

\newcommand{\tr}{\ensuremath{{\sf T}}}
\newcommand{\fa}{\ensuremath{{\sf F}}}
\newcommand{\und}{\ensuremath{{\sf U}}}
\newcommand{\eqp}{=_{\mathsf{p}}}
\newcommand{\nep}{\ne_{\mathsf{p}}}

\newcommand{\eqc}{=}
\newcommand{\existsp}{\exists_{\mathsf{p}}}
\newcommand{\forallp}{\forall_{\mathsf{p}}}

\newcommand{\true}{\ensuremath{\textsf{true}}}
\newcommand{\false}{\ensuremath{\textsf{false}}}
\newcommand{\undefi}{\ensuremath{\textsf{undefined}}}

\newtheorem{theorem}{Theorem}[subsection]
\newtheorem{la}[theorem]{Lemma}  
\newtheorem{prn}[theorem]{Proposition}  
\newtheorem{cvn}[theorem]{Convention}
\newtheorem{dfn}[theorem]{Definition}

\newcommand{\rat}{\mathbb{Q}}

\newcommand{\leftand}{~
     \mathbin{\setlength{\unitlength}{.9ex}
     \begin{picture}(1.6,1.8)(-.4,0)
     \put(-.8,0){$\wedge$}
     \put(-.66,-0.2){\textcolor{white}{\circle*{0.7}}}
     \put(-.66,-0.2){\circle{0.7}}
     \end{picture}
     }}
     
\newcommand{\fulland}{~
     \mathbin{\setlength{\unitlength}{.9ex}
     \begin{picture}(1.6,1.8)(-.4,0)
     \put(-.8,0){$\wedge$}
     \put(-.66,-0.2){\circle*{0.7}}
     \end{picture}
     }}

\newcommand{\leftor}{~
     \mathbin{\setlength{\unitlength}{.9ex}
     \begin{picture}(1.6,1.8)(-.4,0)
     \put(-.8,0){$\vee$}
     \put(-.71,1.7){\textcolor{white}{\circle*{0.7}}}
     \put(-.71,1.7){\circle{0.7}}
     \end{picture}
     }}

\newcommand{\leftimp}{~
     \mathbin{\setlength{\unitlength}{1ex}
     \begin{picture}(1.8,1.8)
     \put(-.0,0){$\rightarrow$}
     \put(-.1,0.57){\circle{0.6}}
     \end{picture}
     ~}}

\newcommand\twoheaddownarrow{\mathrel{\rotatebox[origin=c]{-90}{$\twoheadrightarrow$}}}

\newcommand{\Sigpm}{\ensuremath{\Sigma^{pd}_m}}
\newcommand{\Lsfolsig}{\ensuremath{\mathsf{L_{sfol}}(\Sigma)}}
\newcommand{\Lsfolpm}{\ensuremath{\mathsf{L_{sfol}}(\Sigpm)}}
\newcommand{\ULsig}{\ensuremath{\mathsf{UL_{sfol}}(\Sigma)}}
\newcommand{\ULpm}{\ensuremath{\mathsf{UL_{sfol}}(\Sigpm)}}
\newcommand{\EqLsig}{\ensuremath{\mathsf{Eq(L_{sfol}}(\Sigma))}}
\newcommand{\EqLpm}{\ensuremath{\mathsf{Eq(L_{sfol}}(\Sigpm))}}

\newcommand{\psitr}{\ensuremath{\psi_{\mathsf{true}}}}
\newcommand{\psifa}{\ensuremath{\psi_{\mathsf{false}}}}
\newcommand{\FTCpmc}{\ensuremath{\FTCpm^{\mathsf{cl}}}}

\begin{document}

\thispagestyle{empty}

\title{Fracterm Calculus for Partial Meadows}

\author{
\small Jan A. Bergstra \& Alban Ponse\\
\small Informatics Institute, University of Amsterdam,\\
\small Science Park 900, 1098 XH, Amsterdam, The Netherlands\\
\small j.a.bergstra@uva.nl 
 \& a.ponse@uva.nl
}

\date{}

\maketitle

\begin{abstract}
Partial algebras and partial data types are discussed with the use of signatures that allow partial functions, and
a three-valued short-circuit (sequential) first order logic with a Tarski semantics.
The propositional part of this logic is also known as McCarthy calculus and has been studied extensively. 

 Axioms for the fracterm calculus of partial meadows are given. The case is made that in this way a 
rather natural formalisation of fields with division operator is obtained. 
It is noticed that the logic thus obtained cannot express that 
division by zero must be undefined.

 An interpretation of the three-valued sequential logic into $\bot$-enlargements of partial algebras is
given, for which it is concluded that the consequence relation of the former logic is semi-computable,
and that the $\bot$-enlargement of a partial meadow is a common meadow.
\\[2mm]
\end{abstract}

{\small\fontsize{8.4}{4}\selectfont \tableofcontents}

\section{Introduction}
\label{sec:1}
Following~\cite{BergstraT2007}, we use \emph{meadow} for fields equipped with an inverse or a division function. 
A \emph{partial} meadow is a field equipped with a partial division function, here written $\frac{x}{y}$ 
in fractional notation. A partial meadow is a partial algebra because division is 
partial as $\frac x 0$ is undefined. A \emph{fracterm} is an expression of the form 
\(
 \frac{p}{q},
\)
where $p$ and $q$ are terms over the signature under consideration.

Fracterm calculus for partial meadows (FTCpm)  fully accepts $\frac{1}{0}$ as a fracterm,  in spite of 
it having no value. At the same time a three valued logic is adopted. 
For instance in FTCpm it is not the case that $\frac{1}{0} = \frac{1}{0}$ and
neither is it the case that $ \frac{1}{0} \neq \frac{1}{0}$. 
Both assertions are considered syntactically valid, however. We will use 
\emph{partial equality}, written
$\eqp$ for the equality sign for partial data types. 
If either $t$ or $r$ has no value then $t\eqp r$
has no classical truth value, i.e. \true\ or \false. 
The assertion $\phi(x)$ with
\[
\phi(x) \equiv x \nep  0 \to \frac{x}{x} \eqp 1
\]
is taken for a fundamental fact about partial meadows which holds for all $x$. 
Contemplating the substitution $x=0$ suggests that the most plausible reading of the implication 
in $\phi(x)$ is a short-circuit implication (also referred to as sequential implication, or as the 
implication of McCarthy's logic), a state of affairs  which we wish to make explicit in the symbolic 
notation. Following~\cite{BergstraBR1995} we will denote short-circuit implication with the connective 
$\leftimp$, thereby changing $\phi$ to $\phi'$:
\[
\phi^\prime (x) \equiv x \nep  0 \leftimp \frac{x}{x} \eqp 1.
\]
We notice  that evaluation of $\phi^\prime(0)$ will not involve an attempt to evaluate the fracterm 
$\frac{0}{0}$, because $0 \nep  0$ evaluates to \false, thereby rendering the implication 
at hand \true\ without further inspection of its consequent. 
Moreover, once division is modelled as a partial function, assigning a Boolean truth value to 
$\frac{0}{0} \eqp \frac{0}{0}$ becomes artificial and the use of a non-classical logic with three or 
more truth values becomes plausible if not unavoidable. One may object that $\frac{0}{0} \eqp \frac{0}{0}$ 
is an assertion that ``asks'' for division by $0$ which for that reason must be excluded. 
The reasons for understanding logical connectives in a short-circuited manner (i.e.\ McCarthy logic) 
are these: the most plausible alternatives are versions of strong Kleene logic which allow 
$\tr \vee x = x \vee \tr = \tr$ for all truth values $x$ including non-Boolean ones
and \tr\ a constant for truth. 
The use of such logics is plausible in computer science, for instance in~\cite{JonesM1994} the equality 
relation $\eqp$ (in~\cite{JonesM1994} simply denoted with $=$) is combined with a strong 
Kleene logic. 
The preference of~\cite{JonesM1994} for that logic is based on the 
symmetry of conjunction and disjunction, symmetries which are lost in the short-circuit case. 
By consequence one finds that $\psi(x) \equiv \frac{x}{0} = 2 \vee x = 3$ 
would be true for $x=3$. Evaluating $\psi(3)$, however, seems to call for an 
attempt to evaluate $ \frac{3}{0}$ which is impossible for elementary arithmetic. 
For these reasons we deviate from~\cite{JonesM1994} by adopting short-circuit logic
as preferable in the case of rational numbers with a partial division function. 
Some general background on short-circuit logic is given in Section~\ref{sec:6}.

\medskip

We are unaware of any existing work which investigates the details of the above 
view of partial meadows, or stated differently of formalising arithmetics involving a 
partial division function.

\subsection{Survey of the paper}
We consider the following results of this paper to be new:
\begin{enumerate}[(i)]
\setlength\itemsep{-.4mm}
\item
The claim that short-circuit logic is the most natural logic for use in the context of  partial meadows, 
plus a listing of reasons, together constituting a rationale, for that claim 
(in Secion~1.2 below).
\item
A complete axiomatisation of the fracterm calculus of partial meadows. Here we notice that two variants of completeness must be distinguished: 
\\[-6mm]
\begin{itemize}
\setlength\itemsep{-.4mm}
\item
Axiomatisation completeness: complete Axiomatisation of a class of (partial) algebras using a logic with given semantics.
\item
Logical completeness: completeness of a proof system for a logic with given semantics.
\\[-6mm]
\end{itemize}
Axiomatisation completeness is the focus of our paper.
Axiomatisation completeness (with fracterm calculus for partial meadows as the intended application) can be 
understood and appreciated without any regard to logical completeness. 
We intend this paper to be fully self-contained, both mathematically and in terms of motivation, and to be 
independent from any proof theoretic considerations.
\item
A result on fracterm flattening in the context of the fracterm calculus of partial meadows.
\item
An interpretation of the short-circuit logic for partial algebras into the first order logic of 
algebras enlarged with an element $\bot$ which serves as an absorptive element. The interpretation 
demonstrates that the consequence relation for short-circuit logic on partial algebras is semi-computable. 
Moreover, said interpretation indirectly provides a proof system, or rather a formal proof method for short-circuit 
logic over partial algebras, thereby demonstrating that a proof theory can be developed in principle.
As a proof of concept, it is shown that the $\bot$-enlargement of a partial meadow is a \emph{common} meadow, see
e.g.~\cite{BergstraP2021,BergstraP2016,BergstraT2022CJ} for common meadows.
\end{enumerate}

\subsection{On the rationale of fracterm calculus for partial meadows}
Fracterm calculus for partial meadows aims at formalising elementary arithmetic including division.
Modelling division as a partial function is motivated by the idea 
that this is the most conventional viewpoint on division. All views where division has been made total are, 
however useful these views may be in the context of certain specific objectives, somehow unconventional.

Even if formalisation of elementary arithmetic can be simplified by having division total (which may be 
achieved in different ways), then it is still the case that using a partial function for modelling 
division provides a very relevant alternative which merits systematic investigation. 
At this stage it is unclear whether or not, and if so to which extent,  totalisation of division is helpful 
for the formalisation and understanding of elementary arithmetic. 
Once the notion that division is modelled as a partial function, undefined when the 
denominator equals zero, has been adopted, it is plausible to require  (i.e.\ as a design 
constraint for texts) that texts are to be written in such a manner and ordering that a reader, 
when reading in the order of presentation, will never be asked or invited to contemplate 
what happens when dividing by 0.

\section{Partial algebras with 3-valued sequential logic}
\label{sec:2}

In Section~\ref{sec:2.1} we define a three-valued sequential first order logic for partial algebras, and
in Section~\ref{sec:2.2} we discuss its semantics. 

\subsection{Sequential first order logic for partial algebras}
\label{sec:2.1}

We start by recalling an equational logic that defines the sequential, short-circuited connectives
${\leftand}$ and ${\leftor}$,
called \emph{short-circuit} (or \emph{sequential}) disjunction and conjunction.
In~\cite{BP23} we introduced so-called $\CLSCL$, Conditional Short-Circuit Logic, together with
a relatively simple semantics based on evaluation trees for the sequential evaluation of atoms 
(propositional variables). 
In this paper, atoms will be partial equalities.

$\CLSCL$ is equivalent with three-valued Conditional Logic, as introduced by Guzmán and Squier in~\cite{GS90},
but is distinguished by the use of the specific notation for the short-circuit connectives mentioned above. 
In~\cite{BP23} we provided several complete, independent
equational axiomatisations of $\CLSCL$, based on whether or not to include the constants 
\tr, \fa\ and \und\ for the three truth 
values \true, \false\ and \undefi, respectively. 

\begin{table}[t]
\caption{The set \CLe\ of axioms of $\CLSCL$ with \tr, \fa\ and the connective $\leftimp$}
\label{tab:CL2}
\vspace{1.6mm}
\centering
\hrule
\begin{align*}
\label{Neg}
\tag{e1}
&
&\fa
&=\neg\tr
\\
\label{Or}
\tag{e2}
&&\phi\leftor \psi
&=\neg(\neg \phi\leftand\neg \psi)
\\
\label{Tand}
\tag{e3}
&&\tr\leftand \phi
&=\phi
\\
\label{Abs}
\tag{e4}
&&\phi\leftand(\phi\leftor \psi)
&=\phi
\\
\label{Mem}
\tag{e5}
&&(\phi\leftor \psi)\leftand \xi
&=(\neg \phi\leftand(\psi\leftand \xi))\leftor(\phi\leftand \xi)
\\
\label{Com}
\tag{e6}
&&(\phi\leftand \psi)\leftor(\psi\leftand \phi)
&=(\psi\leftand \phi)\leftor(\phi\leftand \psi)
\\
\tag{e7}
\label{Imp}
&&\phi\leftimp \psi
&=\neg \phi\leftor \psi
\end{align*}
\hrule
\end{table}

In Table~\ref{tab:CL2}, we define the set \CLe\ of axioms of $\CLSCL$ with constants \tr\ and \fa\ and
and with the addition of the connective $\leftimp$ that defines short-circuit implication.
The axioms of \CLe\ imply double negation elimination, i.e.
\begin{equation}
\label{DNS}
\tag{DNE}
\neg\neg \phi=\phi,
\end{equation}
and therewith a sequential version of
the duality principle: in equations without occurrences of $\leftimp$, the connectives
$\leftand$ and $\leftor$ and occurrences of \tr\ and \fa\ can be swapped. 
Axiom~\eqref{Abs} is a sequential version of the absorption law. Axioms~\eqref{Mem} and ~\eqref{Com}
imply some other properties of these connectives and~\eqref{Imp} defines ${\leftimp}$.
Next to~\eqref{DNS}, other useful consequences of \CLe\ are the following:
\begin{enumerate}[(a)]
 \setlength\itemsep{-.4mm}
\item 
\label{(I)}
$\phi\leftand \tr=\phi$ and $\tr\leftor \phi=\tr$,
\item 
\label{(II)}
$(\phi\leftand \psi)\leftand \xi=\phi\leftand(\psi\leftand \xi)$, so the connective $\leftand$ is associative, 
\item 
\label{(III)}
$\phi\leftand(\psi\leftand \phi)=\phi\leftand \psi$, so, with $\psi=\tr$, $\leftand$ is idempotent, 
i.e. $\phi\leftand \phi=\phi$,  
\item
\label{(IV)}
$\phi\leftand (\psi\leftor \xi)=(\phi\leftand \psi)\leftor (\phi\leftand \xi)$, so  
$\leftand$ is left-distributive, 
\item
\label{(V)}
$\phi\leftand \neg \phi=\neg \phi\leftand \phi$, 
\item
\label{(VI)}
$(\phi\leftand \psi)\leftimp \xi=\phi\leftimp(\psi\leftimp \xi)$.
\end{enumerate}
Of course, the duals of \eqref{(I)}--\eqref{(V)}, say \eqref{(I)}$'$--\eqref{(V)}$'$, are also consequences of \CLe.
Consequences~\eqref{DNS} and \eqref{(I)}--\eqref{(VI)} follow quickly with help of
the theorem prover \emph{Prover9}~\cite{Prover9}, see Appendix~\ref{A.0}. 
Simple short proofs of~\eqref{DNS} and~\eqref{(I)}--\eqref{(V)} from (only)  \eqref{Neg}--\eqref{Mem}
are included in~\cite[App.A.3 \& A.4]{BPS21}.
At the end of Section~\ref{sec:2.2} (and also in Appendix~\ref{A.0}) we discuss an example that refutes 
the commutativity of ${\leftor}$ and give a completeness result for \CLe.

Based on \CLe, we define a sequential first order logic for partial
algebras, and in Section~\ref{sec:3} we wil refine this logic to the specific case of partial meadows. 
Let $\Sigma$ be a signature that contains one or more partial functions
and a constant symbol $c$.
Consider a partial $\Sigma$-algebra $A$ with non-empty domain $|A|$. 
There is an equality relation called \emph{partial equality} and written $\eqp$ which behaves on 
all domains as usual, i.e. as a normal 
equality relation. This means that for $a,b \in |A|$, $a \eqp b$ if, and only if, $a = b$, where 
$=$ is the normal equality relation on $|A|$. 
For clarification, we will illustrate further axioms and rules of inference with informal use of 
the semantics defined in Section~\ref{sec:2.2}.

With \Lsfolsig\ we denote the collection of formulae inductively defined as follows,
assuming a set $V_{var}$ of variables:
\begin{itemize}
 \setlength\itemsep{-.0mm}
\item 
$\tr,\fa \in \Lsfolsig$,
\item 
for $\Sigma$-terms $t$ and $r$ the \emph{atomic} formulae
$(t \eqp r) \in \Lsfolsig$,
\item
if $\phi \in \Lsfolsig$ then $ \neg \phi \in  \Lsfolsig$,
\item
if $\phi,\psi \in \Lsfolsig$, then  
$\phi \leftor \psi$, 
$\phi \leftand \psi$, and 
$\phi \leftimp \psi$ are in $\Lsfolsig$,
\item
if $\phi \in \Lsfolsig$ and $x \in V_{var}$, then  
$\existsp x. \phi$ and $\forallp x. \phi$ are in \Lsfolsig.
\end{itemize}
So, in \Lsfolsig, there is explicit quantification over variables that occur in atomic formulae.
For example, we will see that a partial meadow satisfies 
\[(\forallp x.x\nep 0 \leftimp \frac xx\eqp 1)=\tr.\]
We extend \CLe\ by adding the following axiom: 
\begin{equation}
\label{e8}
\tag{e8}
\existsp x.\phi = \neg\forallp x.\neg\phi.
\end{equation}

Next, we define \ULsig, the universal quantifier-free fragment of \Lsfolsig, and \EqLsig,
the collection of equations over \Lsfolsig\ and the largest language we consider:
\begin{itemize}
 \setlength\itemsep{-.mm}
\item If $\phi\in \Lsfolsig$ is quantifier-free, then $\phi\in\ULsig$,
\item If $\phi,\psi \in \Lsfolsig$ then $(\phi=\psi)\in\EqLsig$.
\end{itemize}
Hence, in both \ULsig\ and \EqLsig, variables 
can occur free.
We give an example of a valid \EqLsig-formula: 
a partial meadow $A$ satisfies both $(x\nep 0\leftimp\frac xx\eqp 1)=\tr$ and
$(\forallp x.x\nep 0\leftimp\frac xx\eqp 1)=\tr$, 
which is characterised by the equivalence
\[
A\models(x\nep 0\leftimp\tfrac xx\eqp 1)=\tr \iff A\models(\forallp x.\: x\nep 0\leftimp\tfrac xx\eqp 1)=\tr.  
\]
More generally, for a partial $\Sigma$-algebra $A$,
the following equivalences hold true for $\phi\in\ULsig$:
\begin{align}
\label{qF}
\tag{qF}
&A\models\phi = \tr \iff 
A\models(\forallp x.\phi) = \tr\iff A\models(\neg\existsp x.\neg\phi)=\tr,
\end{align}
where the second equivalence is a consequence of axiom~\eqref{e8}. 
The related equivalence $A\models(\existsp x.\phi) = \tr\iff A\models(\neg\forallp x.\neg\phi)=\tr$ 
is more complex, see Section~\ref{sec:2.2}.

\smallskip

We are especially interested in $\EqLsig$-formulae of the form
\[\phi=\tr,\]
and give below axioms and rules for deriving such equations, accompanied by  
examples for partial meadows (which will be formally defined in Section~\ref{sec:3}).
When proving results, and in some cases when writing axiom systems or designing operations on the syntax, 
we assume for simplification that all formulas can be written with a subset of the sequential connectives, 
for example, with connectives $\neg$ and $\leftor$ and with the universal quantifier $\forallp$ only
(by axiom~\eqref{e8}). 
Alternatively one may assume that connectives ${\leftand}$ and ${\leftimp}$ are used as well as both quantifiers
$\existsp$ and $\forallp$, 
on top of equations and denial inequations, 
that is negated equations written in the form $t \nep  r$.
\\[1.97mm]
\textbf{Axioms for the relation $\eqp$.} 
This relation is not a congruence relation because it is not reflexive on expressions, 
for example, in a partial meadow it is not the case that $\frac 10 \eqp \frac 10$ and neither is 
it the case that $\frac 10\nep\frac 10$.
Equations of the form $t\eqp t$ are used to express that $t$ is defined
and occur in the `weak substitution
property' (p6) and (p7) (cf.\ \cite{BergstraBTW81}).
We write  $\leftand_{i=1}^1\phi_i = \phi_1$ and
for $k>0$, $\leftand_{i=1}^{k+1} \phi_i=\phi_1\leftand(\leftand_{i=2}^{k+1}\phi_i)$. 

Axioms (and axiom schemes) for the relation $\eqp$ are the following, for any constant $c\in\Sigma$
and $x\in V_{var}$ and 
(open) $\Sigma$-terms $t_i$:
\\[1.58mm]
\begin{tabular}[t]{ll}
(p1)
&$(c\eqp c)=\tr$, 
\hfill(reflexivity for constants)
\\[1.58mm]
(p2)
&$(x\eqp x)=\tr$,
\hfill(reflexivity for variables)
\\[1.58mm]
(p3)
&$((\leftand_{i=1}^2 t_i\eqp t_i)\!\leftimp\! (t_1\eqp t_2~\leftimp~ t_2\eqp t_1))=\tr$,
\hfill(symmetry) 
\\[1.58mm]
(p4)
&$((\leftand_{i=1}^3 t_i\eqp t_i)\!\leftimp\!((t_1\eqp t_2\leftand t_2\eqp t_3)\leftimp t_1\eqp t_3))=\tr$.
\quad\hspace{10mm}
(transitivity)
\end{tabular}
\\[1.58mm]
For each $k$-ary $f\in\Sigma$ that is total
and all (open) $\Sigma$-terms $t_1,...,t_k$:
\\[1.58mm]
$\begin{tabular}[t]{ll}
\textup{(p5)}
&$((\leftand_{i=1}^k t_i\eqp t_i)
\leftimp
f(t_1,...,t_k)\eqp f(t_1,...,t_k))=\tr$.
\quad\hspace{22.9mm}
(definedness)
\end{tabular}
$
\\[1.58mm]
\quad\hspace{22.9mm}
For all $k$-ary $f\in\Sigma$
and all (open) $\Sigma$-terms $t_1,...,t_k, ~r_1,...,r_k$, the  \emph{weak substitution property}:
\\[1.58mm]
$\begin{array}[t]{rlrr}
\textup{(p6)}
&(f(t_1,...,t_k)\eqp f(t_1,...,t_k)\leftimp
 (\leftand_{i=1}^k t_i\eqp t_i))&=\tr,
\\[2.42mm]
\textup{(p7)}
&(((f(t_1,...,t_k)\eqp f(t_1,...,t_k)\leftand (\leftand_{i=1}^k t_i\eqp r_i)))
&\leftimp ~
\\
&&\multicolumn{2}{r}{
\!\hspace{-36pt}f(t_1,..., t_k) \eqp f(r_1,..., r_k))=\tr.}
\end{array}
$
\\[1.58mm]
For example, if in a partial meadow  $\tfrac tt$ is defined then so is $t$ by (p6),
and if $\tfrac tt$ is defined and $t\eqp r$, then $\tfrac tt\eqp \tfrac rr$ by (p7).
\\[1.58mm]
\noindent
\textbf{More axioms for $\EqLsig$.}
For $t$ a $\Sigma$-term, $x\in V_{var}$, $\phi\in\Lsfolsig$, and constant $c\in\Sigma$:
\\[1.58mm]
\begin{tabular}[t]{ll} 
(a1)
&$((t\eqp t\leftand\forallp  x.\phi) ~\leftimp~ \phi[t/x])=\tr$ 
~if no variables in $t$ occur bounded in $\phi$, 
\\[1.58mm]
(a2)
&$(\forallp x.\:x\eqp c \leftor x\nep c)=\tr$.
\end{tabular}
\\[1.58mm]
For example, with respect to a partial meadow, axiom (a1)
excludes the substitution $x\mapsto \frac y0$ (because that 
would introduce undefinedness). A more extensive example using (a2) is the derivation
of~\eqref{Ex:1} on page~\pageref{Ex:1}.
Note that by consequence \eqref{(VI)} of \CLe, (a1) can
also be written as 
\[(t\eqp t\leftimp(\forallp x.\phi \leftimp \phi[t/x]))=\tr.\]

Below we define two replacement rules (i1) and (i2) with 
help of the notion of a `context': if a formula
$\phi$ in $\Lsfolsig$ has a subformula $\psi$, then $\phi$ is a \emph{context}
for $\psi$, notation
\(\phi\equiv C[\boldsymbol\psi],\)
where $\psi$ can be written in boldface to indicate a single occurrence. 
For example, 
if 
$\phi\equiv(\psi\leftand\xi)\leftor(\neg\psi)$, then
\[
\phi\equiv
C_1[\boldsymbol\psi]\equiv(\boldsymbol\psi\leftand\xi)\leftor(\neg\psi) 
~~\text{and}~~
\phi\equiv
C_2[\boldsymbol\psi]\equiv(\psi\leftand\xi)\leftor(\neg\boldsymbol\psi).
\]
\textbf{Rules of inference for $\EqLsig$}, where we write 
$\vdash \phi=\psi$ for the derivability
of $\phi=\psi$ from the above axioms, in particular including those of \CLe\ (Table~\ref{tab:CL2}) 
extended with axiom~\eqref{e8}:
\\[1.78mm]
\begin{tabular}[t]{rl}
(i1)& 
~If $\vdash C[\boldsymbol{t\eqp r}]=\tr$ and $\vdash (r\eqp s)=\tr$   
then $\vdash C[\boldsymbol{t\eqp s}]=\tr$,
\hspace{1.8mm}({replacement} rule 1)
\\[1.78mm]
(i2)& 
~If $\vdash C[\boldsymbol{\phi}]=\tr$ and  $\vdash\phi = \psi$ 
then $\vdash C[\boldsymbol{\psi}]=\tr$,
~\hfill({replacement} rule 2)
\\[1.78mm]
(i3)&
~If $\vdash \phi=\tr$ and  $\vdash(\phi\leftimp\psi) = \tr$ then $\vdash \psi=\tr$,
\hfill({modus ponens} for $\leftimp$)
\\[1.78mm]
(i4)&
~If $\vdash \phi =\tr$ and $\vdash \psi=\tr$, then $\vdash (\phi\leftand \psi)=\tr$,
\hfill
($\leftand$-introduction)
\\[2.12mm]
(i5)&
~If $\vdash(\phi\leftimp\psi)=\tr$ and $\vdash(\psi\leftor\neg\psi)=\tr$ and $\vdash(\xi\leftor\neg\xi)=\tr$,
\\[1.08mm]
&~then $\vdash((\phi\leftor\xi)\leftimp(\psi\leftor \xi))=\tr$.
\hfill($\leftor$-introduction)
\end{tabular}
\\[1.58mm]
Inference rule (i1) is a replacement rule that captures some properties of $\eqp$.
For example, a partial meadow has the properties that $\vdash (1\eqp 1+0)=\tr$ and $\vdash(0\nep 1)=\tr$.
So, with $C[{0\eqp 1}]=\neg({0\eqp 1})$  it follows by 
(i1) that $\vdash(0\nep 1+0)=\tr$. 

As for (i2), we first note that this rule
expresses the common congruence rule for \CLe-expressions.
Furthermore, $\phi$ is a subformula of itself ($\phi\equiv C[\phi]$), 
so if $\vdash \phi=\psi$ and $\vdash\phi=\tr$, then $\psi\equiv C[\psi]$ and $\psi=\tr$ is derived by (i2). 
Below, we give an example that uses (i2).

For an example of (i3), instantiate $\phi$ with $\forallp x.\:x\eqp c \leftor x\nep c$, 
i.e.\ the left-hand side of (a2):
$\vdash (x\eqp x\leftimp((\forallp x.\:x\eqp c \leftor x\nep c)\leftimp (x\eqp c\leftor x\nep c)))=\tr$.
By (p1), $\vdash (x\eqp x)=\tr$, so by (i3), 
\[\vdash ((\forallp x.\:x\eqp c \leftor x\nep c)\leftimp 
(x\eqp c\leftor x\nep c))=\tr.\]
With (a2) and (i3), the leftmost consequence of \eqref{Ex:1} below 
is obtained, and hence, with consequence $\eqref{(V)}'$ 
(i.e. $\vdash \phi_1\leftor \neg\phi_1=\neg\phi_1\leftor\phi_)$, say $\vdash \phi=\psi$) and (i2) 
with $\phi\equiv C[\phi]$, also the second consequence of \eqref{Ex:1}, i.e. $C[\psi]=\tr$:
\begin{equation}
\label{Ex:1}
\tag{Ex.1}
\vdash (x\eqp c\leftor x\nep c)=\tr,\quad\text{and thus}\quad
\vdash (x\nep c\leftor x\eqp c)=\tr.
\end{equation}
As for rule (i5), we note that 
the conditions $\vdash(\psi\leftor\neg\psi)=\tr$ and 
$\vdash(\xi\leftor\neg\xi)=\tr$ exclude undefinedness of $\psi$ and $\xi$;
for an application of (i5), see the proof of Theorem~\ref{CFTF4PM}.
Inference rules (i2)--(i5) are derivable from $\CLe$, this follows easily with
\emph{Prover9}~\cite{Prover9}.

\begin{la}
\label{la:nuttig}
The relation $\{(t,r)\mid\text{$t,r~\Sigma$-terms such that $\vdash (t\eqp t)=\tr$}, \\
\vdash (r\eqp r)=\tr \text{~and~} \vdash(t\eqp r)=\tr\}$ is a congruence.
\end{la}

\begin{proof}
The combination of (p5) and (p7) implies with (i4) that for each total $k$-ary function $f$ and 
for all defined $\Sigma$-terms $t_1,...,t_k$ and $r_1,...,r_k$
that satisfy $\vdash (t_i\eqp r_i)=\tr$, $(f(t_1,..., t_k) \eqp f(r_1,..., r_k))=\tr$.
So, with (p1)--(p4) we are done.
\end{proof}

\subsection{\smash{Inductively defined Tarski semantics of $\EqLsig$}}
\label{sec:2.2}

Assume $A$ is a partial $\Sigma$-algebra and $\sigma$ is a valuation that assigns values from 
$|A|$ to the variables in $V_{var}$. 
We write $\sigma[a/x]$ with $a\in|A|$ for the valuation that assigns $a$ to $x$ and is otherwise defined as
$\sigma$.
A valuation $\sigma$ extends to $\Sigma$-terms provided their constituents are defined: 
if 
$f(\sigma(t_1),...,\sigma(t_k))$ is defined, then 
$\sigma(f(t_1,...,t_k))=f(\sigma(t_1),...,\sigma(t_k))$. 

\medskip

We first consider formulae of the form $\phi=\tr$ (and their universal quantifications)
and define $A,\sigma \models \phi=\tr$ if, and only if, 
\(
A,\sigma \models \phi.
\)
A key role is played by the negation of partial equality $A,\sigma \models t \nep r$, also called denial 
inequality.
For denial inequality we will have three notations: 
\[\text{$A,\sigma \models \neg (t \eqp r)$, $A,\sigma \models t \nep r$, and 
$A,\sigma \notmodels t \eqp r$.}
\]
We notice that when working with total algebras, $\eqp$ coincides with $=$ and 
denial inequality coincides with dissatisfaction: 
$A,\sigma \models t \nep r \iff 
A,\sigma \not \models t \eqp r$.

\medskip

For $\phi \in \Lsfolsig$, satisfaction $A,\sigma \models \phi$ complemented with  
denial satisfaction  $A,\sigma \notmodels \phi$  
is inductively defined as follows:
\\[-6mm]
\label{complemented}
\begin{enumerate}[ (1)]
 \setlength\itemsep{-.4mm}
\item 
$A,\sigma \models\tr$,

\item 
$A,\sigma \notmodels\fa$,

\item $A,\sigma \models t\eqp r $ if, and only if, both evaluation results $\sigma(t)$
and $\sigma(r)$ are defined (exist) and are equal: $\sigma(t) = \sigma(r)$,
 
\item $A,\sigma \notmodels t\eqp r $ if, and only if, both evaluation results $\sigma(t)$
and $\sigma(r)$ are defined (exist) and are different elements of $|A|$:
$\sigma(t) \ne \sigma(r)$,
  
\item   $A,\sigma \models \neg \phi$ if  $A,\sigma \notmodels \phi$,

\item   $A,\sigma \notmodels \neg \phi$ if  $A,\sigma \models \phi$, 

\item
\label{cd5}
$A,\sigma \models \phi_1 \leftor \phi_2$ if either  (i) $A,\sigma \models \phi_1$ or 
(ii) $A,\sigma \notmodels \phi_1$ and  $A,\sigma \models \phi_2$,
 
\item   $A,\sigma \notmodels \phi_1 \leftor \phi_2$ if   $A,\sigma  \notmodels  \phi_1$ and 
 $A,\sigma \notmodels \phi_2$,

\item
$A,\sigma\models\forallp  x. \phi$  if for all $a \in |A |$, 
it is the case that $A,\sigma[a/x] \models  \phi$,

\item
$A,\sigma\notmodels\forallp  x. \phi$  if for some $a \in |A |$, 
it is the case that $A,\sigma[a/x] \notmodels  \phi$
and for all $b \in |A |$, it is the case that $A,\sigma[b/x] \models  \phi\leftor\neg\phi$
~(i.e. $\phi$ is not undefined).

\end{enumerate}
It follows that $A,\sigma\models\phi\leftimp\psi$ if either (i) $A,\sigma\notmodels\phi$
or (ii) $A,\sigma \models \phi$ and  $A,\sigma \models \psi$.
By duality and inference rule~(i2), it follows that
$A,\sigma\models\phi\leftand\psi$ if $A,\sigma\models\phi$ and $A,\sigma\models\psi$.
This does not contradict the fact that ${\leftand}$ and ${\leftor}$ are \emph{not} commutative,
as is demonstrated at the end of this section. 
Furthermore, it follows that
\begin{align}
\tag{qE}
\label{eq:forall}
\text{
\begin{tabular}[b]{l}  
$A,\sigma \models \existsp  x. \phi$ if for some $a \in |A |$, it is the case that 
$A,\sigma[a/x] \models \phi$ 
\\
\quad
and for all $b \in |A |$, it is the case that
$A,\sigma[b/x] \models \phi\leftor\neg\phi$.
\end{tabular}}
\\[2mm]
\nonumber
\text{\begin{tabular}[t]{l}
$A,\sigma \notmodels \existsp  x. \phi$ if for all $a \in |A |$, it is the case that 
$A,\sigma[a/x] \notmodels  \phi$.
\end{tabular}}
\end{align}
Note that in a quantification $\forallp x. \phi$ or 
$\existsp x. \phi$ the free variable $x$ ranges over $|A|$ and not over $\Sigma$-terms. So a
term $t$ may only be substituted for the variable $x$ when the definedness of $t$ is
valid.

\medskip

As usual, with 
\[A   \models \phi\]
it is denoted that for all valuations $\sigma$ it is the case that 
$A  ,\sigma \models \phi$. With $A \notmodels \phi$ it is denoted that for all valuations $\sigma$, 
$A,\sigma \notmodels \phi$.

\begin{prn}
\label{Prop:soundness}
The axioms and inference rules for 
$\EqLsig$ given in Section~\ref{sec:2.1} are valid 
in all partial $\Sigma$-algebras.
\end{prn}

\begin{proof}[Proof (sketch)]
Here we only consider the validity of axioms~\eqref{Neg}--\eqref{Imp} (Table~\ref{tab:CL2}) and~\eqref{e8}
restricted to their \true-cases, and we will comment on this point at the end of this section. Furthermore, only  
the less trivial cases are considered.
Let a partial $\Sigma$-algebra $A$ and valuation $\sigma$ be given.
\\[1.28mm]
\underline{Axiom~\eqref{Mem}}, i.e., $(\phi\leftor \psi)\leftand \xi
=(\neg \phi\leftand (\psi\leftand \xi))\leftor (\phi\leftand \xi)$, say $L=R$. Assume that
 $A,\sigma\models L$. This is the case if $A,\sigma\models \phi\leftor \psi$ and $A,\sigma\models \xi$,
so, by \eqref{cd5}-(i), $A,\sigma\models \phi$ and $A,\sigma\models \xi$ (subcase 1), or, by \eqref{cd5}-(ii), 
$A,\sigma\notmodels \phi$ and $A,\sigma\models\psi$ and $A,\sigma\models\xi$ (subcase 2).
Then $A,\sigma\models R$: in subcase~1, $A,\sigma\notmodels\neg \phi\leftand (\psi\leftand \xi)$ 
because $A,\sigma\models \phi$, and since $A,\sigma\models \xi$ it follows that 
$A,\sigma\models \phi\leftand\xi$.
In subcase 2, it follows immediately that $A,\sigma\models\neg \phi\leftand (\psi\leftand \xi)$.
If $A,\sigma\models R$ then $A,\sigma\models L$ follows in a similar way.
\\[1.28mm]
\underline{Axiom~\eqref{Com}}, 
i.e. $(\phi\leftand \psi)\leftor (\psi\leftand\phi)
=(\psi\leftand\phi)\leftor (\phi\leftand \psi)$, say $L=R$. Assume that
$A,\sigma\models L$. This is the case if $A,\sigma\models \phi\leftand \psi$, 
which is the case if
$A,\sigma\models \phi$ and $A,\sigma\models \psi$. (The other ${\leftor}$-case, i.e.  $A,\sigma \notmodels \phi\leftand \psi$ and 
$A,\sigma\models \psi\leftand \phi$, cannot occur).
Hence $A,\sigma\models R$.
If $A,\sigma\models R$ then, by symmetry, $A,\sigma\models L$.
\\[1.28mm]
\underline{Axiom~\eqref{e8}},
i.e. $\existsp x.\phi = \neg\forallp x.\neg\phi$.
Assume that
$A,\sigma\models\neg\forallp x.\neg\phi$, then $A,\sigma\notmodels\forallp x.\neg\phi$, hence by \eqref{eq:forall},
for some $a \in |A |$, it is the case that
$A,\sigma[a/x] \notmodels \neg\phi$ and thus $A,\sigma[a/x] \models \phi$, and for all 
$b \in |A |$, it is the case that $A,\sigma[b/x] \models \neg\phi\leftor\phi$ and thus by consequence~\eqref{(VI)}, $A,\sigma[b/x] \models \phi\leftor\neg\phi$, which is the defining clause for $A,\sigma\models\existsp x.\phi$.
\\[1.28mm]
\underline{Axiom (a1)}, i.e. $((t\eqp t\leftand\forallp  x.\phi)\leftimp \phi[t/x])=\tr$, or equivalently,
$(t\eqp t\leftimp(\forallp  x.\phi ~\leftimp~ \phi[t/x]))
\\=\tr$, and assume no variables in $t$ occur bounded in $\phi$. If $A,\sigma\models t\eqp t$
and $A,\sigma\models \forallp  x.\phi$, then $t$ is defined, and if $\sigma(t)=b\in |A|$, 
then $A, \sigma[b/x]\models \phi$, which implies $A,\sigma\models \phi[t/x]$. 
\\[1.28mm]
\underline{Axiom (a2)}, i.e. $(\forallp x.\:x\eqp c \leftor x\nep c)=\tr$.
We have to show that for all $a\in|A|$, 
$A,\sigma[a/x]\models x\eqp c\leftor x\nep c$. If $a=\sigma(c)$, then $A,\sigma[a/x]\models\textup{(a2)}$, 
and otherwise, $A,\sigma[a/x]\models x\nep c$, hence $A,\sigma\models\textup{(a2)}$. 
\\[1.28mm]
\underline{The  rules of inference}. Rule (i1) is obviously sound, and inference rules (i2)--(i5) 
follow quickly with \emph{Prover9}~\cite{Prover9}, see Appendix~\ref{A.0}. 
\end{proof}

\medskip

A partial equality $t\eqp r$ is {undefined} if at least one of $t$ and $r$ is undefined.
In the following we generalise `undefinedness' to $\Lsfolsig$-formulae and 
$\EqLsig$-formulae using an 
adaptation of the notation for satisfiability. For a partial $\Sigma$-algebra $A$ and a valuation $\sigma$,
we write
\[
A,\sigma \unmodels \phi
\]
to express that $\phi$ is undefined in $A$ according to  $\sigma$, and we define this notion inductively:
\\[-6mm]
\begin{enumerate}[(1)]
 \setlength\itemsep{-.4mm}
\item 
$A,\sigma \unmodels t\eqp r $ if, and only if, at least one of both evaluation results $\sigma(t)$
and $\sigma(r)$ is undefined,

\item 
\label{(11)} 
$A,\sigma \unmodels \neg \phi$ if $A,\sigma \unmodels \phi$, 
 
\item 
$A,\sigma \unmodels \phi_1 \leftor \phi_2$  if either  (i)
$A,\sigma  \unmodels  \phi_1$ or (ii)
$A,\sigma \notmodels \phi_1$ and  $A,\sigma \unmodels \phi_2$,

\item 
$A,\sigma \unmodels \forallp  x. \phi$  if for some $a \in |A |$, it is the case that 
$A,\sigma[a/x] \unmodels \phi$.

\end{enumerate}
Some consequences that follow from clause~\eqref{(11)} and axiom~\eqref{e8}
(i.e. $\forallp x. \phi = \neg\existsp  x. \neg \phi$):
\begin{align}
\nonumber
&A,\sigma \unmodels \phi_1 \leftand \phi_2 
&&\text{  iff~ either  (i)
$A,\sigma  \unmodels  \phi_1$ or (ii) $A,\sigma \models \phi_1$ and $A,\sigma \unmodels \phi_2$},
\\
\label{c-1}
\tag{e8-1}
&A,\sigma \unmodels \forallp  x. \neg\phi 
&&\text{  iff~ $A,\sigma \unmodels \forallp  x. \phi$},
\\
\label{c-2}
\tag{e8-2}
&A,\sigma \unmodels \existsp  x. \phi 
&&\text{  iff~ $A,\sigma \unmodels \forallp  x. \phi$}.
\end{align}
It easily follows that for any formula $\phi\in\ULsig$  (thus $\phi$ quantifier-free)
and any $\sigma$,
either $A,\sigma\models \phi$, or $A,\sigma\notmodels \phi$, or $A,\sigma\unmodels \phi$.

We write $A\unmodels \phi$ if for all valuations $\sigma$, $A,\sigma\unmodels \phi$.
Some examples where we take $A$ to satisfy
$0,1 \in|A|$,  
$A\models 0\nep 1$, 
$A\unmodels \frac 10\eqp t$ for any term $t$
and $A\models \frac 11\eqp 1$
 (properties of a partial meadow):
\begin{enumerate}[(i)]
 \setlength\itemsep{.8mm}
\item 
\label{(i)}
$A \unmodels \frac 10\eqp 1$, so $A\unmodels\frac10\eqp 1\leftor 0\nep 1$, 
while $A \models 0\nep 1\leftor\frac10\eqp 1$, so ${\leftor}$ is not commutative.

\item
\label{(ii)}
$A\unmodels \forallp  x. \frac xx\eqp 1$ because $A\unmodels \frac 00\eqp 1$, 
and thus also $A \unmodels \existsp  x. \frac xx\eqp 1$.
\end{enumerate}
More generally, if $A,\sigma \unmodels \forallp  x. \phi$ then neither $A,\sigma \models \forallp  x. \phi$
nor $A,\sigma \models \forallp  x. \neg\phi$. This follows by choosing $a\in |A|$ that witnesses
$A,\sigma[a/x] \unmodels \phi$ (and hence also $A,\sigma[a/x] \unmodels \neg\phi$).

For $\EqLsig$-formulae and a partial $\Sigma$-algebra  $A$ we define 
$A,\sigma\models \phi=\psi$ by the following three clauses:
\\[-6mm]
\begin{enumerate}[(1)]
 \setlength\itemsep{-.4mm}
\item either $A,\sigma\models \phi$ and $A,\sigma\models \psi$,
\item or $A,\sigma\notmodels \phi$ and $A,\sigma\notmodels \psi$,
\item or $A,\sigma\unmodels \phi$ and $A,\sigma\unmodels \psi$.
\end{enumerate}
We write $A\models \phi=\psi$ if $A,\sigma\models \phi=\psi$ for all $\sigma$.
For example, if $A$ is a partial meadow (like in \eqref{(i)}--\eqref{(ii)} above), we find 
$A\models (\frac 10\eqp 1)=(\frac10\eqp 0)$. For another example, see Proposition~\ref{prop:X}.

The following result concerns the axiomatisation $\CLe$ in Table~\ref{tab:CL2} and is a minor 
generalisation of the result cited in~\cite{BP23} because short-circuit implication $\leftimp$ has 
been added as a definable connective. It immediately follows that this addition preserves this result.
\\[2mm]
\textbf{Theorem.} (See~\cite[Thm.5.2.(ii)]{BP23}.)
\emph{Conditional logic with \tr\ and $\:$\fa\ distinguished 
is completely axiomatised by the seven axioms ~\eqref{Neg}--\eqref{Imp} of \CLe\ in Table~\ref{tab:CL2}.
Moreover, these axioms are independent.}

\begin{proof}[\textbf{A comment on this theorem and the proof sketch of 
Proposition~\ref{Prop:soundness}}]
~In the proof sketch of Proposi\-tion~\ref{Prop:soundness} we did not take into account that 
a partial $\Sigma$-algebra $A$ allows the interpretation of formulas of the form $\phi=\psi$ in which both 
$\phi$ and $\psi$ can be undefined. 
For example, a partial meadow $A$ satisfies 
$A\models (x\leftand(\frac10\eqp\frac10))\leftor(\frac10\eqp\frac10)= (\frac10\eqp\frac10)$. 
By using a constant \und\ to represent the truth value \undefi\ and the corresponding axiom $\neg\und=\und$, 
it indeed follows that $\CLe\cup\{\neg\und=\und\}\vdash (x\leftand\und)\leftor\und=\und$ (see~\cite{BP23}). 
Finally, we note that axiom~\eqref{e8}, i.e. $\existsp x.\phi = \neg\forallp x.\neg\phi$, also covers all cases involving undefinedness. This follows from
the defining clauses that for any valuation $\sigma$,
$A,\sigma\unmodels \neg\phi$ if $A,\sigma\unmodels \phi$ and 
$A,\sigma \unmodels \forallp  x. \phi$  if for some $a \in |A |$, 
$A,\sigma[a/x] \unmodels \phi$, and from consequences\eqref{c-1} and~\eqref{c-2} above.
\end{proof}

\section{Fracterm calculus for partial meadows}
\label{sec:3}

In Section~\ref{sec:3.1} we define `fracterm calculus for partial meadows', FTCpm in short, and 
prove a completeness result for FTCpm.
In Section~\ref{sec:3.2}, we introduce a convention for concise notation for FTCpm for the sake of  
readability of axioms and proofs.
In Section~\ref{sec:3.3} we derive some properties of partial meadows and prove some results, among which 
``conditional flattening'', and pay attention to some $\EqLsig$-identies.

\subsection{FTCpm: a specification}
\label{sec:3.1}
The signature of partial meadows with divisive notation $\Sigma_{m}^{pd}$ is obtained by extending the signature of 
unital rings with a two place division operator (denoted $\frac{x}{y}$), where it is indicated in the signature 
description that division is a partial function. 
The sort of numbers involved is named $\mathsf{Number}$. 

\begin{dfn} 
A \textbf{partial meadow} is a structure $F^{pd} $ with signature $\Sigma_{m}^{pd}$ that is obtained 
by expanding a field $F$ with a partial division operator (with the usual definition, 
i.e. $\frac ab=c$ if $b\ne 0$ and
$b\cdot c=a$).
\end{dfn}

In $\smash{\Sigma_{m}^{pd}}$ constants are supposed to have a value and functions, except division, are supposed
to be total.

Starting from our definition in Section~\ref{sec:2.1} of a sequential first order logic,
we understand `fracterm calculus for partial meadows', FTCpm in short, as the collection of
$\ULpm$-formulae $\phi$ (thus $\phi$  quantifier-free), 
such that $\phi$ is  valid in all partial meadows, i.e.\ for each partial meadow $F^{pd}$,
\[F^{pd} \models \phi\] 
(that is, for all valuations $\sigma$, $F^{pd} ,\sigma \models \phi$). 

\begin{table}[t]
\caption{Specification of the signature $\Sigma^{pd}_m$ of fracterm calculus of partial meadows and 
\\
\phantom{Table~\ref{tab:FTCpm}\hspace{1.4mm}}
 a set $\FTCpm$ of axioms in the format of \EqLpm
} 
\label{tab:FTCpm}
\vspace{1.6mm}
\centering
\hrule
\begin{align}
\nonumber	
\nonumber	
\mathsf{signature}~&\colon {\Sigma_{wcr,\bot}}=\{  
\\
\nonumber
\texttt{sort}~&\colon \mathsf{Number}
\\
\nonumber
\mathsf{constants}~&\colon 0,1,\bot\colon  \mathsf{Number}
\\
\nonumber
\mathsf{total~functions}~&\colon \_+\_\,,\_\cdot\_\,\colon 
\mathsf{Number} \times  \mathsf{Number} \to \mathsf{Number};
\\
\nonumber
&~~-\_ \colon \mathsf{Number} \to  \mathsf{Number}
\\
\nonumber
\mathsf{equality~relation}~&\colon \_  =\_ \subseteq  \mathsf{Number} \times \mathsf{Number} \}
\\
\nonumber
\hspace{26mm}\mathsf{variables}~&\colon x,y,z \colon  \mathsf{Number}
\end{align}
\vspace{-8.4mm}
\begin{align}
\label{pm1a}
\tag{pm1a}
((x+y)+z \eqp x + (y + z))&=\tr
\hspace{22mm}
\\[1mm]
\label{pm2a}
\tag{pm2a}
(x+0 \eqp x)&=\tr
\\[1mm]
\label{pm3a}
\tag{pm3a}
(x + (-x) \eqp 0)&=\tr 
\\[1mm]
\label{pm4a}
\tag{pm4a}
(x \cdot (y \cdot z) \eqp (x \cdot y) \cdot z)&=\tr
\\[1mm]
\label{pm5a}
\tag{pm5a}
(x \cdot y \eqp y \cdot x)&=\tr
\\[1mm]
\label{pm6a}
\tag{pm6a}
(1 \cdot x \eqp x)&=\tr
\\[1mm]
\label{pm7a}
\tag{pm7a}
(x \cdot (y+z) \eqp (x \cdot y) + (x \cdot z))&=\tr 
\\[1mm]
\label{pm8a}
\tag{pm8a}
(y \nep  0  
\leftimp
\frac{x}{y} 
\eqp x  \cdot  \frac 1y)&=\tr
\\[1mm] 
\label{pm9a} 
\tag{pm9a}
(x \nep  0  
\leftimp
\frac{x }{x}  \eqp 1)&=\tr  
\\[1mm] 
\label{pm10a}
\tag{pm10a}
(0\nep 1)&=\tr
\\[2mm]
\label{pm11a}
\tag{pm11a}
((x\nep 0\leftand y\nep 0)
\leftimp x\cdot y\nep 0)&=\tr
\end{align}
\hrule
\end{table}

A specification of $\Sigma_{m}^{pd}$ with a set $\FTCpm$ of axioms for FTCpm is given 
in Table~\ref{tab:FTCpm}.
Observe that $(\frac 11\eqp 1)=\tr$ follows immediately  from axiom~\eqref{pm9a}, 
which in turn with axiom~\eqref{pm8a} implies
\(
(\frac x1\eqp x)=\tr.
\)
We note that axioms~\eqref{pm10a} and~\eqref{pm11a} capture common properties of a field: 
$0\ne 1$ and absence of zero divisors.

\begin{theorem}[Soundness and completeness of $\FTCpm$]
\label{thm:completeness}
~\textup{
\begin{enumerate}[(i)]
 \setlength\itemsep{0.8mm}
\item \textit{(Soundness) The axioms of $\FTCpm$ are sound for the class of partial meadows, and}
\item \textit{(Completeness) If a universal formula $\phi \in \ULpm$ 
is true in all partial meadows,
then $\FTCpm \models \phi$.}
\end{enumerate}}
\end{theorem} 

\begin{table}[t]
\caption{The set $\FTCpm$ of axioms according to the first equivalence of \eqref{qF} (see page \pageref{qF})} 
\label{tab:FTCpm2a}
\vspace{1.6mm}
\centering
\hrule
\begin{align}
\nonumber	
\mathsf{import}~&\colon \Sigma^{pd}_m~(\mathrm{Table}~\ref{tab:FTCpm})
\\
\nonumber
\mathsf{variables}~&\colon x,y,z \colon  \mathsf{Number} 
\end{align}
\vspace{-8.4mm}
\begin{align}
\label{pm1c}
\tag{pm1b}
(\forallp x.\forallp y.\forallp z.\:(x+y)+z \eqp x + (y + z))&=\tr
\hspace{12mm}
\\[1mm]
\label{pm2c}
\tag{pm2b}
(\forallp x.\:x+0 \eqp x)&=\tr
\\[1mm]
\label{pm3c}
\tag{pm3b}
(\forallp x.\:x + (-x) \eqp 0)&=\tr 
\\[1mm]
\label{pm4c}
\tag{pm4b}
(\forallp x.\forallp y.\forallp z.\:x \cdot (y \cdot z) \eqp (x \cdot y) \cdot z)&=\tr
\\[1mm]
\label{pm5c}
\tag{pm5b}
(\forallp x.\forallp y.\:x \cdot y \eqp y \cdot x)&=\tr
\\[1mm]
\label{pm6c}
\tag{pm6b}
(\forallp x.\:1 \cdot x \eqp x)&=\tr
\\[1mm]
\label{pm7c}
\tag{pm7b}
(\forallp x.\forallp y.\forallp z.\:x \cdot (y+z) \eqp (x \cdot y) + (x \cdot z))&=\tr 
\\[1mm]
\label{pm8c}
\tag{pm8b}
\forallp x.\forallp y.(y \nep  0  
\leftimp
\frac{x}{y} 
\eqp x  \cdot  \frac 1y)&=\tr
\\[1mm] 
\label{pm9c} 
\tag{pm9b}
(\forallp x.\:x \nep  0  
\leftimp
\frac{x }{x}  \eqp 1)&=\tr  
\\[1mm] 
\label{pm10c}
\tag{pm10b}
(0\nep 1)&=\tr
\\[2mm]
\label{pm11c}
\tag{pm11b}
(\forallp x.\forallp y.(x\nep 0\leftand y\nep 0)
\leftimp x\cdot y\nep 0)&=\tr
\end{align}
\hrule
\end{table}

\begin{proof} Soundness is obvious by inspection of the axioms of $\FTCpm$. For completeness assume 
that $\phi$ is valid in all partial meadows. Now consider a model $B$ of $\FTCpm$.
$B$ is an expansion with a division function of a ring and in fact of a field. $B$ may differ from 
a partial meadow because the division function may be defined for some pairs of arguments $(a,b)$ 
with $b=0$. Now let $B'$ be obtained from $B$ by replacing the division function of $B$ by the 
standard partial division function. $B'$ is a partial meadow and therefore $B' \models \phi$. 
It can be shown by induction on the structure
of open formulae $\psi$ that for all valuations $\sigma$
it is the case that $B',\sigma \models \psi$ implies 
$B \models \psi$. We find that each model of $\FTCpm$ satisfies $\phi$, as required.
\end{proof}

\subsection{FTCpm: concise notations}
\label{sec:3.2}
To enhance readability, we introduce the following convention. 
\begin{cvn} 
\label{convention}
\mbox{\hspace{1mm}}
\\[-6mm]
\begin{itemize}
 \setlength\itemsep{.6mm}
\item[$(1)$]
We adopt the convention of writing 
$\phi$ instead of $\phi=\tr$.
 
\item[$(2)$]
$\forallp$-elimination: for expressions of the form $\forall_p x.\phi$ (with $\phi \in \Lsfolpm$) 
we adopt the convention of writing $\phi$. 
\vspace{.6mm}
\end{itemize}
\noindent
If not explicitly mentioned, $\phi$ here either stands for a syntactic formula (an element of $\Lsfolpm$), or for
$\,\vdash \phi$, or for
$\,\models \phi$, and this should then always be clear from the context.
\end{cvn}

\begin{table}[t]
\caption{Representation of the axioms of $\FTCpm$ according to 
Convention~\ref{convention}.(1)} 
\label{tab:FTCpm2b}
\vspace{1.6mm}
\centering
\hrule
\begin{align}
\nonumber	
\mathsf{import}~&\colon \Sigma^{pd}_m~(\mathrm{Table}~\ref{tab:FTCpm})
\\
\nonumber
\mathsf{variables}~&\colon x,y,z \colon  \mathsf{Number} 
\end{align}
\vspace{-8.4mm}
\begin{align}
\label{pm1C}
\tag{pm1c}
\forallp x.\forallp y.\forallp z.\:
(x+y)+z &\eqp x + (y + z)
\hspace{20mm}
\\[1mm]
\label{pm2C}
\tag{pm2c}
\forallp x.\:
x+0 &\eqp x
\\[1mm]
\label{pm3C}
\tag{pm3c}
\forallp x.\:
x + (-x) &\eqp 0 
\\[1mm]
\label{pm4C}
\tag{pm4c}
\forallp x.\forallp y.\:
x \cdot (y \cdot z) &\eqp (x \cdot y) \cdot z
\\[1mm]
\label{pm5C}
\tag{pm5c}
\forallp x.\forallp y.\:
x \cdot y &\eqp y \cdot x
\\[1mm]
\label{pm6C}
\tag{pm6c}
\forallp x.\:
1 \cdot x &\eqp x
\\[1mm]
\label{pm7C}
\tag{pm7c}
\forallp x.\forallp y.\:
x \cdot (y+z) &\eqp (x \cdot y) + (x \cdot z) 
\\[1mm]
\label{pm8C}
\tag{pm8c}
\forallp x.\forallp y.\:
y \nep  0  
\leftimp
\frac{x}{y} 
&\eqp x  \cdot  \frac 1y 
\\[1mm] 
\label{pm9C} 
\tag{pm9c}
\forallp x.\:
x \nep  0  
\leftimp
\frac{x }{x} 
&\eqp 1  
\\[1mm] 
\label{pm10C}
\tag{pm10c}
0
&\nep 1
\\[2mm]
\label{pm11C}
\tag{pm11c}
\forallp x.\forallp y.\:
(x\nep 0\leftand y\nep 0)
\leftimp x\cdot y&\nep 0
\end{align}
\hrule
\end{table}

\begin{table}[t]
\caption{Representation of the set $\FTCpm$ of axioms according to Convention~\ref{convention}} 
\label{tab:FTCpm2}
\vspace{1.6mm}
\centering
\hrule
\begin{align}
\nonumber	
\mathsf{import}~&\colon \Sigma^{pd}_m~(\mathrm{Table}~\ref{tab:FTCpm})
\\
\nonumber
\mathsf{variables}~&\colon x,y,z \colon  \mathsf{Number} 
\end{align}
\vspace{-8.4mm}
\begin{align}
\label{pm1}
\tag{pm1}
(x+y)+z &\eqp x + (y + z)
\\[1mm]
\label{pm2}
\tag{pm2}
x+0 &\eqp x
\\[1mm]
\label{pm3}
\tag{pm3}
x + (-x) &\eqp 0 
\\[1mm]
\label{pm4}
\tag{pm4}
x \cdot (y \cdot z) &\eqp (x \cdot y) \cdot z
\\[1mm]
\label{pm5}
\tag{pm5}
x \cdot y &\eqp y \cdot x
\\[1mm]
\label{pm6}
\tag{pm6}
1 \cdot x &\eqp x
\\[1mm]
\label{pm7}
\tag{pm7}
x \cdot (y+z) &\eqp (x \cdot y) + (x \cdot z) 
\\[1mm]
\label{pm8}
\tag{pm8}
y \nep  0  
\leftimp
\frac{x}{y} 
&\eqp x  \cdot  \frac 1y 
\\[1mm] 
\label{pm9} 
\tag{pm9}
x \nep  0  
\leftimp
\frac{x }{x} 
&\eqp 1  
\\[1mm] 
\label{pm10}
\tag{pm10}
0
&\nep 1
\\[2mm]
\label{pm11}
\tag{pm11}
(x\nep 0\leftand y\nep 0)
\leftimp x\cdot y&\nep 0
\end{align}
\hrule
\end{table}
 
In Tables~\ref{tab:FTCpm2a}--\ref{tab:FTCpm2}, we display the axioms of $\FTCpm$ in
different, equivalent formats.
Table~\ref{tab:FTCpm} represents the $\forallp$-elimination of the axioms of Table~\ref{tab:FTCpm2a}, and Table~\ref{tab:FTCpm2} shows the $\forallp$-elimination of the axioms of Table~\ref{tab:FTCpm2b}.
Of course, Convention~\ref{convention} 
is justified by the Tarski-semantics discussed in Section~\ref{sec:2.2}. 

As an example,
we give for axiom~\eqref{pm9a} of Table~\ref{tab:FTCpm} a diagram that illustrates 
this {convention},
using the first equivalence of~\eqref{qF},
i.e. 
\(
F^{pd}\models\phi = \tr\iff F^{pd}\models(\forallp x.\phi) = \tr
\)
(see page~\pageref{qF}), in which the prefix ``$F^{pd}\models$ '' is omitted:
\\[3.14mm]
\mbox{~}\hspace{4mm}
\begin{tabular}[t]{ll}
(Tbl.\ref{tab:FTCpm})~~~$(x\nep 0\leftimp \frac xx \eqp 1)=\tr\quad\iff\quad
(\forallp x.x\nep 0\leftimp \frac xx \eqp 1)=\tr$ 
&(Tbl.\ref{tab:FTCpm2a})
\\[3.14mm]
\mbox{~}~~\hspace{24.2mm}
$\twoheaddownarrow
\hspace{58.4mm}
\twoheaddownarrow $
&
\\[3.14mm]
(Tbl.\ref{tab:FTCpm2})
\hspace{3mm}~~$x\nep 0\leftimp \frac xx \eqp 1
\hspace{12.6mm} 
\twoheadleftarrow
\hspace{7.4mm}
\forallp x.x\nep 0\leftimp \frac xx \eqp 1$
&(Tbl.\ref{tab:FTCpm2b})
\end{tabular}
\vspace{3mm}
\\[2.24mm]\indent
Convention~\ref{convention} and these examples relate to both the axioms in Section~\ref{sec:2.1}
and those for partial meadows.
In what follows, we will work mostly according to this convention and the axioms in Table~\ref{tab:FTCpm2}, 
but where practical, we will use more explicit representation. A typical example is the following consequence:
\[
\text{If $t$ is defined, thus $t\eqp t$, and $\forallp x.\phi(x)$ is a valid assertion, 
then so is $\phi(t)$.}
\]
This `instantiation consequence' is justified by axiom (a1), i.e. 
$((t\eqp t\leftand\forallp  x.\phi) ~\leftimp~ \phi[t/x])=\tr$. 
An example of the instantiation consequence applied to axiom~\eqref{pm11} and the (derivable)
identity $-1\eqp-1$ is then
\((-1\nep 0\leftand y\nep 0)\leftimp -1\cdot y \nep 0.
\)

We end this section with a final remark on the axioms for partial meadows.
Observe that the condition $(x\nep 0\leftand y\nep 0)$ of axiom~\eqref{pm11a} in 
Table~\ref{tab:FTCpm} and in its related versions in Tables~\ref{tab:FTCpm2a}--\ref{tab:FTCpm2}
can be replaced by $(y\nep 0\leftand x\nep 0)$.

\subsection{\smash{Some consequences of FTCpm and \EqLpm}}
\label{sec:3.3}

Table~\ref{FTCpm3} lists some familiar consequences (in the sense of 
$\models$) of the axioms of $\FTCpm$ (Table~\ref{tab:FTCpm}). 
Note that 
the order of the conjuncts in Assertion~\eqref{eq:fracfrac} is also not relevant.

\begin{table}[t]
\caption{Assertions of fracterm calculus of partial meadows (with $x,y,u,v\in V_{var}$)} 
\label{FTCpm3}
\vspace{1.6mm}
\centering
\hrule
\begin{align}
\label{pmComm+}
\tag{A1}
x+y &~\eqp~ y+x
\\[1mm]
\label{eq:0}
\tag{A2}
0\cdot x&~\eqp~ 0
\\[1mm] 
\label{eq:-x}
\tag{A3}
x \nep 0 
&~\leftimp~
{-}x			
\nep
0
\\[1mm] 
\label{eq:minus}
\tag{A4}
y \nep  0 
&~\leftimp~
{-}\frac{x}{y}		
\eqp
\frac{-x}y
\\
\label{eq:prod}
\tag{A5}
(y \nep  0  \leftand  v \nep  0) 
&~\leftimp~
\frac{x}{y}	\cdot \frac{u}{v}		
\eqp
\frac{x \cdot u}{y \cdot v}
\\
\label{eq:fracfrac}
\tag{A6}
(y \nep 0  \leftand u \nep 0 \leftand  v \nep 0) 
&~\leftimp~
\dfrac{(\frac{x}{y})}{( \frac{u}{v})}		
\eqp
\frac{x \cdot  v}{y \cdot u}
\\
\label{eq:plus}
\tag{A7}
(y \nep 0  \leftand  v \nep 0 ) 
&~\leftimp~
\frac{x}{y} + \frac{u}{v}	
\eqp
\frac{(x\cdot v) + (y \cdot u)}{y \cdot v}
\end{align}
\hrule
\end{table}

\begin{prn}
\label{prop:pm}
The assertions \eqref{pmComm+}--\eqref{eq:plus} in Table~\ref{FTCpm3} follow from 
$\FTCpm$. 
\end{prn}

\begin{proof}
Adapting to Convention~\ref{convention}, we  apply Lemma~\ref{la:nuttig} tacitly and
start with three auxiliary results:
\begin{enumerate}[(a)]
\setlength\itemsep{-1mm}
\item 
\label{(a)}
\underline{$0 + x \eqp x$}: $0 + -(-x) \eqp (x + -x) + -(-x) \eqp x + (-x + -(-x)) \eqp x + 0 \eqp x$, so 
 $0 + -(-x) \eqp x$, hence
$~ 0 + x \eqp 0 + (0 + -(-x)) \eqp (0 + 0) + -(-x) \eqp 0 + -(-x) \eqp x$.
\item 
\label{(b)}
\underline{$-(-x)\eqp x$}: $-(-x) \stackrel{(a)}\eqp 0 + -(-x) \eqp (x + -x) + -(-x) \eqp x + (-x + -(-x)) = x + 0 \eqp x$.
\item
\label{(c)}
\underline{$-x + x \eqp 0$}:
$-x + x \stackrel{(b)}\eqp -x + -(-x) \eqp 0$.
\end{enumerate}
\underline{Assertion~\eqref{pmComm+}}. 
Derive
$(1+1)\cdot(x+y) \eqp (1+1)\cdot x + (1+1)\cdot y \eqp (x+x) + (y+y)$ and
$(1+1)\cdot(x+y) \eqp (x+y)\cdot(1+1) \eqp (x+y) + (x+y)$. 
So after the left addition of $-x$ and the right addition of $-y$ we find 
with associativity and auxiliary results \eqref{(a)}--\eqref{(c)} that $x+y\eqp y+x$.
\\[1.22mm]
\underline{Assertion~\eqref{eq:0}}. Derive
$0\cdot x \eqp (0+0)\cdot x\eqp 0\cdot x+0\cdot x$, hence $0 \eqp 0\cdot x + (-(0\cdot x))
\eqp (0\cdot x + 0\cdot x) + (-(0\cdot x))
\eqp 0\cdot x + (0\cdot x + (-(0\cdot x)))
\eqp 0\cdot x + 0
\eqp 0\cdot x$.
\\[1.22mm]
\underline{Assertion~\eqref{eq:-x}}.
By instantiation of axiom~\eqref{pm11},
$(-1\nep 0\leftand y\nep 0)\leftimp -1\cdot y\nep 0$.
Derive
\begin{align}
\nonumber
(-x)\cdot y 
&\eqp(-x)\cdot y+(x\cdot y + -(x\cdot y))\eqp((-x)\cdot y+x\cdot y) + -(x\cdot y)
\\
\label{eq:e}
\tag{d1}
&\eqp(((-x)+x)\cdot y) + -(x\cdot y)\stackrel{\eqref{(a)},\eqref{(c)},\eqref{eq:0}}\eqp -(x\cdot y). 
\end{align}
By instantiation and axiom~\eqref{pm6},
$-1\cdot y = -y$, so 
by (i1), $(-1\nep 0\leftand y\nep 0)\leftimp -y\nep 0$.
By (p5) and $0\nep 1$, $-0\nep -1$ and by auxiliary result~\eqref{(c)}, $-0\eqp 0$, hence $0\nep -1$
so by (p3) and consequence~\eqref{(V)} and (i3), $-1\nep 0$.
By 
consequence~\eqref{(V)} and $-1\nep 0$ and (i3), $y\nep 0\leftimp -y\nep 0$.
\\[1.22mm]
\underline{Assertion~\eqref{eq:minus}}. By axiom~\eqref{pm8} and \eqref{eq:e}, 
$\dfrac{-x}y\eqp -x\cdot\dfrac 1y\eqp -(x\cdot\dfrac 1y)\eqp {-}\dfrac xy$.
\\[1.22mm]
\underline{Assertion~\eqref{eq:prod}}.
By~\eqref{pm8} and \eqref{pm6} with $x=1$ we find $y\nep 0\leftimp \frac 1y\eqp\frac 1y$,
and similarly $v\nep 0\leftimp \frac 1v\eqp\frac 1v$. 
By~\eqref{pm11} and~\eqref{pm9}, $(y\nep 0\leftand v\nep 0)\leftimp\frac{y\cdot v}{y\cdot v}\eqp 1$.  
Hence 
\begin{equation}
\label{yeq}
\tag{d2}
\frac{y\cdot v}{y\cdot v}\cdot \frac 1y
\stackrel{\eqref{pm8}}\eqp ((y\cdot v)\cdot \frac 1{y\cdot v})\cdot \frac 1y
\eqp ((y\cdot v)\cdot \frac 1y)\cdot \frac 1{y\cdot v}
\stackrel{\eqref{pm8}}\eqp \frac{(y\cdot v)\cdot \frac 1y}{y\cdot v},
\end{equation}
so 
\begin{align}
\label{xeq}
\nonumber
\frac 1y\cdot\frac 1v
&
\eqp \Big(\frac{y\cdot v}{y\cdot v}\cdot \frac 1y\Big)\cdot\frac 1v
\stackrel{\eqref{yeq}}\eqp \frac{((y\cdot \frac 1y)\cdot v)}{y\cdot v}\cdot\frac 1v
\stackrel{\eqref{pm8}}\eqp \frac{(\frac yy\cdot v)}{y\cdot v}\cdot\frac 1v
\\
\tag{d3}
&
\eqp \frac{v}{y\cdot v}\cdot\frac 1v
\stackrel{\eqref{pm8}}\eqp (v\cdot\frac{1}{y\cdot v})\cdot\frac 1v
\eqp \frac{1}{y\cdot v},
\end{align}
and thus
\[
\dfrac{x}{y}	\cdot \dfrac{u}{v}
\stackrel{\eqref{pm8}}\eqp (x\cdot \dfrac 1y)\cdot (u\cdot\dfrac 1v)
\stackrel{\eqref{xeq}}\eqp (x\cdot u) \cdot\dfrac 1{y\cdot v}
\stackrel{\eqref{pm8}}\eqp \dfrac {x\cdot u}{y\cdot v}.
\vspace{1mm}
\]
\underline{Assertion~\eqref{eq:fracfrac}}. If $y\nep 0$, $u\nep 0$ and $v\nep 0$, then 
$\frac xy\eqp\frac xy, \frac vu\eqp\frac vu$ and $\frac uv\eqp\frac uv$, and
\[
\frac 1{(\frac uv)}
\eqp \frac{(\frac vu)}{(\frac vu)}\cdot \frac 1{(\frac uv)}
\stackrel{\eqref{eq:prod}}\eqp \frac{(\frac vu)}{(\frac {v\cdot u}{u\cdot v})}
\eqp \frac{(\frac vu)}{1}
\eqp \frac vu,
\]
and thus
\[
\dfrac{(\frac{x}{y})}{( \frac{u}{v})}
\stackrel{\eqref{pm8}}\eqp \frac xy\cdot \frac 1{(\frac uv)}
\eqp \frac xy\cdot \frac vu
\stackrel{\eqref{eq:prod}}\eqp \frac{x \cdot  v}{y \cdot u}.
\vspace{1mm}
\]
\underline{Assertion~\eqref{eq:plus}}.
If $y \nep 0$ and $v \nep 0$, then $\frac yy\eqp 1$ and $\frac vv\eqp 1$.
Hence, with~\eqref{pm5}--\eqref{pm7}, \eqref{pm8} and \eqref{eq:prod} we find
\begin{align}
\nonumber
\dfrac xy +\dfrac uv
&\eqp \dfrac{x\cdot v}{y\cdot v}+\dfrac{y\cdot u}{y\cdot v}
\eqp
((x\cdot v)\cdot\dfrac 1{y\cdot v})+((y\cdot u)\cdot \dfrac 1{y\cdot v})
\\
\nonumber
&\eqp((x\cdot v)+(y\cdot u))\cdot \dfrac 1{y\cdot v}
\eqp \dfrac{(x\cdot v)+(y\cdot u)}{y\cdot v}. 
\\[-.46mm]
\nonumber
\tag*{\qedhere}
\end{align}
\end{proof}

\begin{dfn} 
A \textbf{flat fracterm} is an expression of the form $\frac{p}{q}$ 
that contains precisely one occurrence (i.e. the top level occurrence) 
of the division operator, thus $p$ and $q$ are \textbf{division free} terms.
\end{dfn}

\begin{theorem}[Conditional fracterm flattening for partial meadows]
\label{CFTF4PM}
For each term $t$ there is a division free term $s$ and a flat fracterm  
$r$ such that
\\[-6mm]
\textup{\begin{enumerate}[(i)]
\setlength\itemsep{0mm}
\item
\textit{$s \nep 0 \leftimp t \eqp r$ holds in all partial meadows, and}
\item
\textit{$s \eqp 0 \leftand t \eqp t$ does not hold in any partial meadow under any valuation.}
\end{enumerate}
}
\end{theorem} 

\begin{proof}
Construction of $s$ and $r$ by induction on the structure of $t$. 
If $t\in\{0, 1, x\}$ with $x\in V_{var}$, then take $s\equiv 1$ and $r \equiv \frac t1$, and (i) and (ii) follow immediately.
\\[1.22mm]
\underline{If $t\equiv {-}u$}, let division free term $s$ and flat fracterm $r\equiv\frac{n}{d}$ 
be such that (i) and (ii) hold for $u$. 
By soundness  (Thm.\ref{thm:completeness}) and Assertion~\eqref{eq:minus}, 
$s \nep 0 \leftimp  t\eqp \frac{-n}d$ holds in all partial meadows, 
and since $\frac{-n}d$ is a flat fracterm, this proves (i). 

(ii)
If 
$s \eqp 0 \leftand u \eqp u$ 
is not satisfied in any partial meadow under any valuation, then 
this implies that in any partial meadow with a  valuation that satisfies 
$s \eqp 0$, also ${-}u$ is undefined.
Hence 
$s \eqp 0 \leftand {-}u \eqp {-}u$ does not hold in any partial meadow under 
any valuation.
\\[1.22mm]
\underline{If $t \equiv u+v$}, let division free terms $s_u,s_v$ and flat fracterms 
\[
r_u \equiv \dfrac{n_u}{d_u},\quad r_v \equiv \dfrac{n_v}{d_v}\]
be such that $s_u \nep 0 \leftimp u \eqp r_u$ and 
$s_v \nep 0 \leftimp v \eqp r_v$ are true in all partial meadows, while
$s_u \eqp 0 \leftand u \eqp u$ and $s_v \eqp 0 \leftand v \eqp v$ 
are not satisfied in any partial meadow under any valuation.
Now take 
\(
s \equiv s_u \cdot s_v,\quad 
r \equiv \frac{(n_u \cdot d_v)+(d_u \cdot n_v)}{d_u \cdot d_v}.
\)

(i) We first show that  
\begin{equation}
\label{ochoch}
\tag{d4}
\text{$\FTCpm\vdash x\cdot y\nep 0\leftimp x\nep 0$.} 
\end{equation}
By (p7) and \eqref{(V)}, $y\eqp y\leftimp (x\eqp 0\leftimp x\cdot y\eqp 0\cdot y)$. So, by (p2), 
$x\eqp 0\leftimp x\cdot y\eqp 0\cdot y\equiv C[x\cdot y\eqp 0\cdot y]$. 
By Assertion~\eqref{eq:0}, $0\cdot y\eqp 0$, so with (i1), $C[x\cdot y\eqp 0]$, i.e. 
$x\eqp 0\leftimp x\cdot y\eqp 0$.
By inference rule (i5) and the derivability of $x\eqp 0\leftor x\nep 0$ (see \eqref{Ex:1} on page~\pageref{Ex:1})
and of $x\cdot y\eqp 0\leftor x\cdot y\nep 0$, 
\[(x\eqp 0\leftor x\nep 0)\leftimp (x\cdot y\eqp 0\leftor x\nep 0),
\]
so by 
\eqref{Ex:1} and (i3), $x\cdot y\eqp 0\leftor x\nep 0$. Hence $\FTCpm\vdash x\cdot y\nep 0\leftimp x\nep 0$.
By soundness,
$s \nep 0\leftimp s_u \nep 0$ and $s \nep 0\leftimp s_v \nep 0$
hold in all partial meadows.
By Assertion~\eqref{eq:plus}, 
$A\models s \nep 0 \leftimp u+v\eqp r$, 
and since $r$ is a flat fracterm, this proves (i). 

(ii) If 
$s_u \eqp 0 \leftand u \eqp u$ and $s_v \eqp 0 \leftand v \eqp v$ 
are not satisfied in any partial meadow under any valuation,  
this implies by axiom~\eqref{pm11} that in any partial meadow and valuation that satisfies $s\eqp 0$,
at least one of
$s_u \eqp 0$ and $s_v \eqp 0$ is satisfied, which 
implies
that $u+v$ is undefined.
Hence 
$s \eqp 0 \leftand u+v \eqp u+v$ does not hold in any partial meadow under 
any valuation.
\\[1.22mm]
\underline{If $t \equiv u\cdot v$}, let division free terms $s_u,s_v$ and flat fracterms 
\(
r_u \equiv \dfrac{n_u}{d_u},\quad r_v \equiv \dfrac{n_v}{d_v}\)
be such that $s_u \nep 0 \leftimp u \eqp r_u$ and 
$s_v \nep 0 \leftimp v \eqp r_v$ are true in all partial meadows, while
$s_u \eqp 0 \leftand u \eqp u$ and $s_v \eqp 0 \leftand v \eqp v$ 
are not satisfied in any partial meadow under any valuation.
Now take 
\[
s \equiv s_u \cdot s_v,\quad 
r \equiv \dfrac{n_u \cdot n_v}{d_u \cdot d_v}.
\]
and proceed as in the previous case, using Assertion~\eqref{eq:prod}.
\\[1.22mm]
\underline{If $t \equiv \dfrac{u}{v}$}, let division free terms $s_u,s_v$ and flat fracterms 
\[
r_u \equiv \dfrac{n_u}{d_u},\quad r_v \equiv \dfrac{n_v}{d_v}\]
be such that $s_u \nep 0 \leftimp u \eqp r_u$ and 
$s_v \nep 0 \leftimp v \eqp r_v$ are true in all partial meadows, while
$s_u \eqp 0 \leftand u \eqp u$ and $s_v \eqp 0 \leftand v \eqp v$ 
are not satisfied in any partial meadow under any valuation.
Now take 
\[
s \equiv s_u \cdot s_v \cdot n_v,\quad 
r \equiv \frac{n_u \cdot d_v}{d_u \cdot n_v}.
\]
(i) 
By \eqref{ochoch} and soundness
it follows that $s \nep 0\leftimp s_u \nep 0$ 
holds in all partial meadows, and likewise that $s \nep 0\leftimp n_v \nep 0$ 
and $s \nep 0\leftimp n_v \nep 0$ hold in all partial meadows.

From the assumption that $s_u \nep 0 \leftimp u \eqp \frac{n_u}{d_u}$ is true in 
all partial meadows, it follows that $s_u \nep 0 \leftimp d_u \ne_\mathsf{p} 0$
is true in all partial meadows, and similarly, $s_v \nep 0 \leftimp d_v \ne_\mathsf{p} 0$.

Collecting the above, we find that
$s \nep 0\leftimp (d_u \ne_\mathsf{p} 0 \leftand n_v\ne_\mathsf{p} 0\leftand d_v \ne_\mathsf{p} 0)$
holds in all partial meadows. 
By Assertion~\eqref{eq:fracfrac} it follows that 
\[
(d_u\nep 0 \leftand n_v\nep 0\leftand  d_v \nep 0)\leftimp 
\dfrac{(\tfrac{n_u}{d_u})}{(\tfrac{n_v}{d_v})}\eqp \dfrac{n_u\cdot d_v}{d_u\cdot n_v}
\] 
holds in all partial 
meadows. Combining the latter two implications, we find that
\[s \nep 0 \leftimp t \eqp 
\dfrac{n_u\cdot d_v}{d_u\cdot n_v}
\]
holds in all partial 
meadows. 
Because $r_u\equiv\frac{n_u}{d_u}$ and 
$r_v \equiv \frac{n_v}{d_v}$ are flat fracterms, so is $r$.

(ii) If 
$s_u \eqp 0 \leftand u \eqp u$ and $s_v \eqp 0 \leftand v \eqp v$ 
are not satisfied in any partial meadow under any valuation, then 
this implies that in any partial meadow with a  valuation that satisfies 
$s \eqp 0$, at least one of $u$ and $v$ is undefined, which in turn implies
that $\frac uv$ is undefined.
Hence 
$s \eqp 0 \leftand \frac uv \eqp \frac uv$ does not hold in any partial meadow under 
any valuation.
\end{proof}

In $\EqLpm$,
fracterm flattening can be obtained in a more direct manner.
When using the same notation as in Theorem~\ref{CFTF4PM}, except that we write $\frac pq$ for $r$ (thus, 
$p$ and $q$ division free), the following equation is not valid unless $t$ is always defined: 
$\smash{(t \eqp \frac pq) = \tr}$  and also
$\smash{(t \eqp \frac{p \cdot s}{q \cdot s}) = \tr}$ fails. 
To see this complication notice that for instance $(\frac{1}{x} \eqp \frac{1}{x}) = \tr$ fails, 
while $(x \nep 0 \leftimp \frac{1}{x} \eqp \frac{1}{x}) = \tr$ holds.
Hence, we also find the following $\EqLpm$ identities:
\begin{align*}
(s \nep 0 \leftimp x \eqp t) = (s \nep 0 \leftimp x \eqp \tfrac pq),
\\
(s \eqp 0 \leftimp x \eqp t) = (s \eqp 0 \leftimp 0 \eqp \tfrac 10).
\end{align*}

\begin{prn} 
\label{prop:X}
For each term $t$ there is a division free term $s$ and a flat fracterm $\frac{p}{q}$ such that 
the $\EqLpm$-identity
\[(x \eqp t) = (x \eqp \tfrac{p \cdot s}{q \cdot s})\]
is valid in all partial meadows.
\end{prn}

\begin{proof}
Given $t$, let $s$ and $r=\frac pq$ be as found in Theorem~\ref{CFTF4PM}.
Given a partial meadow $\smash{F^{pd}}$ and a valuation $\sigma$, two cases are distinguished:
\[\smash{F^{pd}},\sigma \models s\nep 0 \text{~and~} \smash{F^{pd}},\sigma \models s\eqp 0.\]
In the first case,  
$\smash{F^{pd}},\sigma \models t \eqp \frac{p}{q}$ and  $\smash{F^{pd}},\sigma \models \frac{s}{s}  \eqp 1$
so that $\smash{F^{pd}},\sigma \models t \eqp \frac{p \cdot s}{q \cdot s}$, in the second case, both $t$ and 
$\frac{p \cdot s}{q \cdot s}$ are undefined in $\smash{F^{pd}}$ under valuation $\sigma$.
\end{proof}

In $\EqLpm$, the equation 
$(\frac 10 \eqp \frac 10) = (\frac 10 \nep \frac 10)$ expresses that division is not total.
However, 
this cannot be expressed in $\Lsfolpm$:

\begin{prn}
\label{prop:3.3.5} 
It is impossible to express in \textup{FTCpm}, thus by formulae in $\Lsfolpm$, 
that division is not total.
\end{prn}

\begin{proof} Suppose that  $F^{pd}\models \phi$ for some partial meadow $F^{pd}$, then we may totalise division 
in accordance with the Suppes-Ono convention $\frac  x0 = 0$, thereby obtaining $\mathsf{Tot}_0(F^{pd})$. 
Now we claim that $\mathsf{Tot}_0(F^{pd}) \models \phi$. 
With simultaneous induction on the structure of $\phi$ one easily proves that: 
(i) if
 $F^{pd} \models \phi$ then  $\mathsf{Tot}_0(F^{pd}) \models \phi$ and 
(ii)   if
 $F^{pd} \models \neg \phi$ then  $\mathsf{Tot}_0(F^{pd}) \models \neg \phi$. 
It follows that no formula $\phi$ can distinguish between $F^{pd}$ and  $\mathsf{Tot}_0(F^{pd})$ 
by being true for $F^{pd}$ and not for $\mathsf{Tot}_0(F^{pd})$.
\end{proof}
Two assertions in $\Lsfolpm$ that involve $\existsp$  
and hold in all partial meadows are these:
\begin{align}
\label{eq:Jan1}
x \nep 0 &\leftimp \existsp y. ( x\cdot y \eqp 1),
\\
\label{eq:Jan2}
x \nep 0 &\leftimp \existsp y. (y \nep 0 \leftand x  \eqp \frac 1 y).
\end{align}
Note that \eqref{eq:Jan1} is related to axiom~\eqref{pm11} and that in~\eqref{eq:Jan2}, 
the conjunct $y \nep 0$ cannot be omitted.

In the case of $F^{pd} = \rat^{pd}$, the partial meadow of rationals, we find a computable partial algebra.
A  specification of the abstract partial data type of rationals is given in Table~\ref{FTCpmRats}, 
where the notation $x^2$ in equation~\eqref{4sq} is an abbreviation for $x\cdot x$.

\begin{prn} 
\label{prop:rat}
The axioms in Table~\ref{FTCpmRats} are satisfied in $\rat^{pd}$ and prove each closed equation and inequation
that is true in $\rat^{pd}$.
\end{prn}

\begin{table}[ht]
\caption{$\mathsf{\FTCpm/4sq}$: A specification of the partial meadow of rationals} 
\label{FTCpmRats}
\vspace{1.6mm}
\centering
\hrule
\begin{align}
\nonumber
\mathsf{import}~&\colon \FTCpm~(\mathrm{Table}~\ref{tab:FTCpm2})
\\
\nonumber
\mathsf{variables}~&\colon x,y,z,u \colon  \mathsf{Number} 
\\[3mm]
\label{4sq}
1 + ((x^2 + y^2) + (z^2 + u^2)) &\nep 0
\end{align}
\hrule
\end{table}

\section{$\bot$-Enlargements and consequence relations}
\label{sec:4}
In Sections~\ref{sec:4.1} and~\ref{sec:4.2} we define a notion of `enlargement'
in order to connect partial algebras and their logic to ordinary first order logic with an absorptive 
element $\bot$ that models partiality.
In Sections~\ref{sec:4.3}  
we provide alternative notions of computability for partial 
algebras, though limited to the case of minimal partial algebras. 
We will focus on minimal algebras only and then with much 
simpler definitions.

\subsection{$\bot$-Enlargement and its converse}
\label{sec:4.1}
Let the \emph{absorptive element} $\bot$ be a new constant symbol, intended to represent ``no proper value'' (i.e.\
$t \neq \bot$ corresponds to $t$ being defined) and let the set of first order 
conditional formulas $T_\bot$ contain the assertions that express that functions produce
$\bot$ on any series of arguments involving $\bot$. For instance, for a three place function $f$, 
$T_\bot$ contains 
\[
f(\bot,y,z) = \bot,\quad f(x,\bot,z) = \bot,\quad f(x,y,\bot)=\bot. 
\]

Given a partial algebra $A$, its $\bot$-enlargement $\mathsf{Enl}_\bot(A)$ is obtained by extending the 
domain with a new element, also denoted $\bot$, that serves as the interpretation of $\bot$.
The notation  $\mathsf{Enl}_\bot(A)$  is taken from~\cite{BergstraT2022JLAMP}.

In the opposite direction, given a total algebra $B$ that contains an absorptive element 
$\bot$, and such that $B\models \exists x. x \neq \bot$, the 
operation $\mathsf{Pdt}_\bot$ (partial data type) as introduced in~\cite{BergstraT2022JLAMP} creates a partial algebra 
$\mathsf{Pdt}_\bot(A)$ with $\bot$ removed from the 
domain and each operation which produces $\bot$ on some 
arguments made partial on these arguments. The name $\mathsf{Pdt}_\bot$ 
suggests that the resulting algebra is minimal, a requirement on all data types. 
Such was the intention in~\cite{BergstraT2022JLAMP}. 
We will use the same notation also in the more general case where the resulting structure 
need not be minimal. 

\begin{prn} 
\label{PdtEnl} If $\bot \notin \Sigma(A)$ and $\bot \notin |A|$ then 
$\mathsf{Pdt}_\bot(\mathsf{Enl}_\bot(A))= A$.
\end{prn}

\begin{prn} 
\label{EnlPdt} If $\bot \in \Sigma(A)$ and $\bot \in |A|$ with $\mathsf{card}(|A|)>1$, then 
\(\mathsf{Enl}_\bot(\mathsf{Pdt}_\bot(A))= A.\)
\end{prn}

\subsection{Reformulating the semantics of $\EqLsig$ in first order terms}
\label{sec:4.2}

Assuming $\bot \notin \Sigma$, let
$\Sigma_{\bot}= \Sigma \cup \{\bot\}$. A pair of transformations $\psi_{\mathsf{true}} $ and 
$\psi_{\mathsf{false}}$ translates formulae in $\mathsf{L_{sfol}}(\Sigma)$ 
to first order formulae over $\Sigma_{\bot}$, i.e. to $\mathsf{L_{fol}}(\Sigma_{\bot})$.
The transformation $\psi_{\mathsf{true}}$ translates formulae that are assumed to evaluate to \true,
and $\psi_{\mathsf{false}}$ is used as an auxiliary operator to deal with negation (i.e. 
$\psi_{\mathsf{true}}(\neg\phi) \equiv \psi_{\mathsf{false}}(\phi)$): 
\begin{enumerate}[ (1)]
\setlength\itemsep{0.48mm}

\item 
$\psi_{\mathsf{true}}(\tr)\equiv\tr$  and $\psi_{\mathsf{true}}(\fa)\equiv\fa$,
\item
$\psi_{\mathsf{false}}(\tr)\equiv\fa$ and $\psi_{\mathsf{false}}(\fa)\equiv\tr$,
\item 
$\psi_{\mathsf{true}}(t \eqp r) \equiv t \neq \bot \wedge r \neq \bot \wedge t = r$
(where $x\ne y$ abbreviates $\neg(x=y)$,
\item 
$\psi_{\mathsf{false}}(t \eqp r) \equiv t = \bot \vee r = \bot \vee t \ne r$,
\item 
$\psi_{\mathsf{true}}(\neg\phi) \equiv \psi_{\mathsf{false}}  (\phi)$~
(hence, $\psi_{\mathsf{true}}  (t \nep  r) \equiv \psi_{\mathsf{false}}(t \eqp r)$),
\item 
$\psi_{\mathsf{false}}(\neg\phi) \equiv \psi_{\mathsf{true}}  (\phi)$~ 
(hence, $\psi_{\mathsf{false}}(t \nep  r) \equiv \psi_{\mathsf{true}}(t \eqp r)$),
\item 
$\psi_{\mathsf{true}}(\phi_1 \leftor \phi_2) \equiv \psi_{\mathsf{true}}(\phi_1) \vee
	(\psi_{\mathsf{false}}(\phi_1)  \wedge \psi_{\mathsf{true}}(\phi_2) ) $,	
\item $\psi_{\mathsf{false}}(\phi_1 \leftor \phi_2) \equiv \psi_{\mathsf{false}}(\phi_1) \wedge
	\psi_{\mathsf{false}}(\phi_2)$,
\item $\psi_{\mathsf{true}}(\forallp x.\phi)\equiv \forall x.(x\ne\bot\to\psi_{\mathsf{true}}(\phi))$,
\item $\psi_{\mathsf{false}}(\forallp x.\phi)\equiv \exists x.(x\ne\bot\wedge\psi_{\mathsf{false}}(\phi))
    \wedge \forall x.\psi_{\mathsf{true}}(\phi\leftor\neg\phi)$.

\end{enumerate}
It follows easily that 
\begin{align*}
&\psi_{\mathsf{true}}(\phi_1 \leftand \phi_2)
\equiv
\psi_{\mathsf{true}}(\phi_1) \wedge \psi_{\mathsf{true}}(\phi_2),
\\
&\psi_{\mathsf{false}}(\phi_1 \leftand \phi_2)
\equiv
\psi_{\mathsf{false}}(\phi_1) \vee(\psi_{\mathsf{true}}(\phi_1)\wedge \psi_{\mathsf{false}}(\phi_2)),
\\[2mm]
&\psi_{\mathsf{true}}(\existsp x.\phi) \equiv 
\exists x. (x \neq \bot \wedge \psi_{\mathsf{true}}(\phi))\wedge \forall x.\psi_{\mathsf{true}}(\phi\leftor\neg\phi),
\\
&\psi_{\mathsf{false}}(\existsp x.\phi) \equiv \forall x. (x \neq \bot \to 
\psi_{\mathsf{false}}(\phi)).
\end{align*}
In order to formulate key properties of the operator $\psi_{\mathsf{true}}$, 
a consequence relation must be chosen.

\subsection{Consequence relations}
\label{sec:4.3}
Various consequence relations can be contemplated in the context of 3-valued logics. 
For an extensive discussion of these options we refer to~\cite{KTB88}.
We will consider the so-called \emph{strong validity} consequence relation, notation $\models_{\mathsf{ss}}$,
where both in the assumptions 
and in the conclusion the formulae (sentences)  are considered ``true'' if these are valid under all valuations. 
Alternatively one may consider $\models_{\mathsf{sw}}$, $\models_{\mathsf{ws}}$, and $\models_{\mathsf{ww}}$
where $w$ indicates weak validity, that is for no valuation a formula is false. 
We propose that in the case of elementary arithmetic the use of $\models_{\mathsf{ss}}$ is preferable to 
the three alternatives just mentioned. Below we will write $\models_\Sigma$ instead of 
$\models_{\mathsf{ss}}$ in order to highlight the role and relevance of the signature involved.

\begin{dfn}
Define $\phi_1,...,\phi_n \models_\Sigma  \phi$ if, and only if, for each $\Sigma$-structure  $A$: 
if for all valuations $\sigma$ into $|A|$ it is the case that  
$A,\sigma \models_\Sigma \phi_1$, ...,  $A,\sigma \models_\Sigma \phi_n$, then for all valuations $\sigma$  
into $|A|$, $A,\sigma \models_\Sigma  \phi$.
\end{dfn}

A connection between the various satisfaction relations and the transformations $\psi_\mathsf{true}$ and
$\psi_\mathsf{false}$ is found under some restrictions.

\begin{prn} 
\label{correspondence:1}
For any partial $\Sigma$-algebra $A$, for each $\mathsf{L_{sfol}}(\Sigma)$ formula $\phi$  
and for each valuation $\sigma$ taking values in $|A|$: 
$A ,\sigma \models_\Sigma \phi $ if, and only if, 
\(\mathsf{Enl}_\bot(A),\sigma \models \psi_{\mathsf{true}} (\phi).\)
\end{prn}

\begin{proof} Straightforward by induction on the structure of $\phi$.
\end{proof}

In the following, let $A$ be a partial $\Sigma$-algebra and $B$ a total $\Sigma$-algebra 
with $\bot\in\Sigma(B)$.

\begin{prn} 
\label{correspondence:2}
$A \models_\Sigma \phi $ if, and only if,  $\mathsf{Enl}_\bot(A) \models \psi_{\mathsf{true}} (\phi)$.
\end{prn}
\begin{proof} Immediate using Proposition~\ref{correspondence:1}.
\end{proof}

\begin{prn} 
\label{correspondence:3}
Let $c \in \Sigma( B)$, and assume that $ B \models c \neq \bot$.
Then $\mathsf{Pdt}_\bot( B) \models_\Sigma \phi $ if, and only if, $B \models \psi_{\mathsf{true}} (\phi)$.
\end{prn}

\begin{prn} 
\label{correspondence:4}
$\phi_1, ...,\phi_n \models_\Sigma  \phi$ if, and only if,
\[
T_\bot \cup \{\psi_\mathsf{true} (\phi_1), ..., \psi_\mathsf{true} (\phi_n)\} \cup 
\{\exists x.x \neq \bot\} \models \psi_\mathsf{true} (\phi).
\]
\end{prn}

\begin{proof} For ``if'', assume that $A \models_\Sigma \phi_1, ..., A \models_\Sigma \phi_n$, then by 
Proposition~\ref{correspondence:2}, 
\[
\mathsf{Enl}_\bot(A) \models \psi_\mathsf{true} (\phi_1), ...,
\mathsf{Enl}_\bot(A) \models \psi_\mathsf{true} (\phi_n).\]
Because $A$ has a non-empty domain, $A \models_\Sigma \existsp x. x\neq \bot$. 
Because $\mathsf{Enl}_\bot(A) \models T_\bot$ 
it follows with $T_\bot \cup \{\psi_\mathsf{true} (\phi_1), ..., \psi_\mathsf{true} (\phi_n)\} \models 
\psi_\mathsf{true} (\phi)$ that $\mathsf{Enl}_\bot(A) \models \psi_\mathsf{true} (\phi)$. Now using 
Proposition~\ref{correspondence:2}, $A \models_\Sigma \phi$. 

For the other direction assume that 
$B \models T_\bot \cup \{\psi_\mathsf{true} (\phi_1), ..., \psi_\mathsf{true} (\phi_n)\} 
\cup \{\exists x.x \neq \bot \}$. Then $\mathsf{Pdt}_\bot(B)$ has a non-empty domain so 
that $A =\mathsf{Pdt}(B)$ is well-defined, and with Proposition~\ref{EnlPdt} $B = \mathsf{Enl}_\bot(A)$.
It follows with Proposition~\ref{correspondence:1} that $A \models_\Sigma \phi_1, ..., 
A \models_\Sigma \phi_n$ so that 
$A \models_\Sigma \phi$ from which one obtains $B \models \psi_{\mathsf{true}}(\phi)$ with 
Proposition~\ref{correspondence:2}.
\end{proof}

\begin{prn} 
\label{correspondence:5}
The consequence relation $\phi_1, ...,\phi_n \models_\Sigma  \phi$ is semi-com\-putable.
\end{prn}

\begin{proof} From Proposition~\ref{correspondence:4} it follows that the consequence at hand is 
effectively 1-1 reducible to an instance of consequence from a semi-computable first order theory, 
which is known to be semi-computable.
\end{proof} 

\section{Fracterm calculus for common meadows}
\label{sec:5}

In Section~\ref{sec:5.1}, we recall \emph{common meadows} and a fracterm calculus for these, FTCcm.
In Section~\ref{sec:5.2}, we establish 
that the axiomatisation of FTCcm is equivalent to
the transformation $\psi_{\mathsf{true}}(\FTCpm)$ 
with \FTCpm\ as shown in Table~\ref{tab:FTCpm2b}. 

\subsection{FTCcm, a specification}
\label{sec:5.1}

Fracterm calculus for common meadows, FTCcm, 
starts by involving the absorptive element $\bot$ and by assuming $\frac x0 = \bot$.

Following~\cite{BergstraT2022CJ}, we adopt a modular approach and first consider $\mathsf{Enl}_\bot(R)$, 
the enlargement of a commutative unital ring $R$ with $\bot$, with axioms in Table~\ref{tab:wcr}.

\begin{table}[t]
\caption{Specification of commutative unital $\bot$-rings, with a set $E_{wcr,\bot}$ of axioms}
\label{tab:wcr}
\vspace{1.6mm}
\centering
\hrule
\begin{align}
\nonumber	
\mathsf{signature}~&\colon {\Sigma_{wcr,\bot}}=\{  
\\
\nonumber
\texttt{sort}~&\colon \mathsf{Number}
\\
\nonumber
\mathsf{constants}~&\colon 0,1,\bot\colon  \mathsf{Number}
\\
\nonumber
\mathsf{total~functions}~&\colon \_+\_\,,\_\cdot\_\,\colon 
\mathsf{Number} \times  \mathsf{Number} \to \mathsf{Number};
\\
\nonumber
&~~-\_ \colon \mathsf{Number} \to  \mathsf{Number}
\\
\nonumber
\mathsf{equality~relation}~&\colon \_  =\_ \subseteq  \mathsf{Number} \times \mathsf{Number} \}
\\
\nonumber
\mathsf{variables}~&\colon x,y,z \colon  \mathsf{Number}
\\[2mm]
\label{cm1}
\tag{c1}
(x+y)+z &= x + (y + z)
\\
\label{cm2}
\tag{c2}
x+y &= y+x
\\
\label{cm3}
\tag{c3}
x+0 &= x
\\
\label{cm4}
\tag{c4}
x + (-x) &= 0 \cdot x 
\\[2mm]
\label{cm5}
\tag{c5}
(x \cdot y) \cdot z	&= x \cdot (y \cdot z) 
\\
\label{cm6}
\tag{c6}
x \cdot y &= y \cdot x
\\
\label{cm7}
\tag{c7}
1 \cdot x &= x
\\
\label{cm8}
\tag{c8}
x \cdot (y+ z) 	
&= (x \cdot y) + (x \cdot z) 
\\[2mm]
\label{cm9}
\tag{c9}
-(-x) &= x
\\
\label{cm10}
\tag{c10}
x + \bot &= \bot
\\[2mm]
\label{eq:0xx=0x}
\tag{c11}
0 \cdot (x\cdot x) &= 0 \cdot x 
\end{align}
\hrule
\end{table}

Since $\bot$ is absorptive, it follows that
\(-\bot=\bot+x=\bot\cdot x=\bot,\)
so the familiar ring identities $0\cdot x=0$ and $x+(-x)=0$ are not valid, but the weaker identities
$0\cdot (x\cdot x)=0\cdot x$ and $x+(-x)=0\cdot x$ are.
In~\cite[Thm.2.1]{BergstraT2022CJ}, the set of axioms $E_{wcr,\bot}$ is defined and 
it is shown that for each equation $t=r$ over $\Sigma(R)\cup\{\bot\}$,
\[ 
E_{wcr,\bot}\vdash t=r\iff \mathsf{Enl}_\bot(R)\models t=r.
\]
For example, 
\[
E_{wcr,\bot}\vdash 0\cdot (x+y)=0\cdot (x\cdot y)
\]
(with $0\cdot (x+y) = 0\cdot ((x+y)\cdot(x+y))$ this follows easily).
With \emph{Mace4}~\cite{Prover9}, it quickly follows that the axioms of $E_{wcr}$ are logically independent.
 
\begin{table}[t]
\caption{Specification of the fracterm calculus of common meadows, with a set $\FTCcm$
\\
\phantom{Table~\ref{tab:cm}\hspace{1.4mm}}
 of axioms}
\label{tab:cm}
\vspace{1.6mm}
\centering
\hrule
\begin{align}
\nonumber
\textsf{import}~&\colon E_{wcr,\bot}\setminus\{\eqref{eq:0xx=0x}\}~(\mathrm{Table}~\ref{tab:wcr})
\\
\nonumber	
\mathsf{signature}~&\colon {\Sigma_{md,\bot}^d=\Sigma_{wcr,\bot}\cup~}\{  
\\
\nonumber
\mathsf{total~functions}~&\colon \frac{~^{\_}~}{\_}\, \colon  
\mathsf{Number} \times  \mathsf{Number} \to  \mathsf{Number} \}
\\
\nonumber
\mathsf{variables}~&\colon x,y \colon  \mathsf{Number}
\\[2mm]
\label{rep1}
\tag{cm1}
\frac xy&\eqc x\cdot\frac 1y
\\[1mm]
\label{rep2}
\tag{cm2}
\frac xx&\eqc 1+\frac 0x
\\[1mm]
\label{rep3}
\tag{cm3}
\frac 1{x\cdot y}&\eqc \frac 1x\cdot\frac 1y
\\[1mm]
\label{rep4}
\tag{cm4}
\frac 1{1+(0\cdot x)}&\eqc 1+(0\cdot x)
\\[1mm]
\label{cm15}
\tag{cm5}
\bot &= \frac{1}{0}	
\end{align}
\hrule
\end{table}

\begin{dfn}
\label{def:cm}
A \textbf{common meadow} is an enlargement $F_\bot$ of a field $F$, which results by first extending 
the domain with an absorptive element $\bot$ and then expanding the structure thus obtained 
with a constant $\bot$ (for said absorptive element) and a division function which is made total 
by adopting 
\[\dfrac{x}{0} =\dfrac{x}{\bot} = \dfrac{\bot}{x}= \bot.
\]
\end{dfn}

A common meadow provides arguably the most straightforward way to turn division into a total operator. 
The fracterm  calculus of common meadows (as discussed in~\cite{BergstraP2021} and  in~\cite{BergstraP2016}) 
has many different axiomatisations, see e.g.~\cite{BergstraT2022CJ}. 
Here we combine axioms~\eqref{cm1}--\eqref{cm10} of $E_{wcr,\bot}$ (Table~\ref{tab:wcr})
and axioms~\eqref{rep1}--\eqref{cm15} of Table~\ref{tab:cm} that define division as a total function 
for each structure containing $\bot$ as an absorptive element.
Table~\ref{tab:cm} lists a set $\FTCcm$ of axioms for FTCcm
following the presentation of~\cite{BergstraP2021}, though using fracterms instead of inverse notation
$x^{-1}$ for $\frac 1x$.

\bigskip

With \emph{Prover9}~\cite{Prover9} it quickly follows 
that axiom~\eqref{eq:0xx=0x} is derivable from $\FTCcm$, which, according to \emph{Mace4}~\cite{Prover9},
is a set of independent axioms. 
Also with \emph{Prover9}, the following familiar consequences 
of $\FTCcm$ are quickly derived:
\[\frac x1 =  x, 
~{-}\frac{x}{y} =  \frac{-x }{y } = \frac{x }{-y }, 
~~\frac{x}{y}	\cdot \frac{u}{v} =  \frac{x \cdot u}{y \cdot v}, 
~~\text{and}~
~\frac xy +\frac uv =\frac{x\cdot v+y\cdot u}{y\cdot v}~.\]

The axioms of $\FTCcm$ allow fracterm flattening: each expression can be proven equal to a flat 
fracterm.
This was first shown in~\cite[Prop.2.2.3]{BergstraP2021} 
with inverse notation $x^{-1}$ for $\frac 1x$ (conversely, division can be defined by $\frac xy = x\cdot y^{-1}$).

\subsection{$\bot$-Enlargement: application of $\psi_{\mathsf{true}}()$ to FTCpm}
\label{sec:5.2}

First, we establish some properties of common meadows. 
By Definition~\ref{def:cm},
\begin{align}
\label{eq:cm0}
&
0\ne\bot
~~~\text{and}~~~
1\ne\bot,
\\
\label{eq:cm1}
&{-(\bot})=\bot, ~ x+\bot=\bot+x=\bot, ~ x\cdot\bot=\bot\cdot x=\bot,
~\text{and}~\mbox{ $\dfrac  x\bot=\dfrac \bot x=\bot$},
\\[-.6mm]
\label{eq:cm3}
&\frac x0=\bot.
\end{align}
Moreover, since $\bot$ is absorbing,
it easily follows that
\begin{align}
\label{eq:cm5}
&{-x}=\bot\to x=\bot, ~~x+y=\bot\to (x=\bot\vee y=\bot), \text{~and}
&x\cdot y =\bot\to(x=\bot\vee y=\bot).
\end{align}
The following result implies that our axiomatisation 
$\FTCpm$ of partial meadows (Table~\ref{tab:FTCpm}) is sufficiently strong.
We write $\FTCpmc$ for the axioms of partial meadows as represented in 
Table~\ref{tab:FTCpm2b}, thus with all universal quantifications made explicit, e.g.
\(
\forallp x.\forallp y.\forallp z.(x+y)+z \eqp x+(y+z).
\)

\begin{theorem}
\label{thm:pmeqcm}
With \eqref{eq:cm0}--\eqref{eq:cm5} it follows that~
$\psi_{\mathsf{true}}(\FTCpmc)\vdash \FTCcm$~ and that\\
$~~\psi_{\mathsf{true}}(\FTCpmc)$~
axiomatises a common meadow.
\end{theorem}

\begin{table}
\caption{transformations
of the axioms of \FTCpmc\ (Table~\ref{tab:FTCpm2b}) 
and Assertion~\eqref{pmComm+} 
(Table~\ref{FTCpm3}) 
\\
\phantom{Table~\ref{tab:psipm}\hspace{1.4mm}}
to assertions about common meadows, 
with applications of properties~\eqref{eq:cm0}--\eqref{eq:cm5}}
\vspace{1.6mm} 
\label{tab:psipm}
\centering
\hrule
\begin{align}
\nonumber
\eqref{pm1C}{:}\quad 
&\hspace{-1mm}
\psitr(\forallp x.\forallp y.\forallp z.~(x+y)+z\eqp x+(y+z))
~\leftrightarrow~((x\ne\bot\wedge y\ne\bot\wedge z\ne\bot)\to
\\
\nonumber
&
\quad((x+y)+z\ne\bot\wedge x+(y+z)\ne\bot \wedge(x+y)+z=x+(y+z)))
\\
\label{cmpm1}
&\quad\leftrightarrow((x\ne\bot\wedge y\ne\bot\wedge z\ne\bot)\to (x+y)+z=x+(y+z))
\quad\text{(by~\eqref{eq:cm5})},
\\[1.42mm]
\nonumber
\eqref{pmComm+}{:}\quad 
&\psitr(\forallp x.\forallp y.~x+y\eqp y+x)
\\
\nonumber
&\quad
\leftrightarrow((x\ne\bot\wedge y\ne\bot)\to (x+y\ne\bot\wedge y+x\ne\bot \wedge x+y=y+x))
\\
\label{cmpm2}
&\quad\leftrightarrow((x\ne\bot\wedge y\ne\bot)\to x+y=y+x),
\\[1.42mm]
\nonumber
\eqref{pm2C}{:}\quad 
&\psitr(\forallp x.~x+0\eqp x)
~\equiv~ (x\ne\bot \to( x+0\ne\bot\wedge x\ne\bot \wedge x+0=x))
\\
\label{cmpm3}
&\quad\leftrightarrow (x\ne\bot \to x+0=x),
\\[1.42mm]
\nonumber
\eqref{pm3C}{:}\quad 
&\psitr(\forallp x.~x+(-x)\eqp 0)
\\
\nonumber
&\quad
\equiv (x\ne\bot \to(x+(-x)\ne\bot\wedge 0\ne\bot \wedge x+(-x)=0))
\\
\label{cmpm4}
&\quad\leftrightarrow( x\ne\bot \to  x+(-x)=0 ),
\\[1.42mm]
\nonumber
\eqref{pm4C}{:}\quad 
&\psitr(\forallp x.\forallp y.\forallp z.~(x\cdot y)\cdot z\eqp x\cdot (y\cdot z))
~\leftrightarrow~((x\ne\bot\wedge y\ne\bot\wedge z\ne\bot)\to\\
\nonumber
&\quad((x\cdot y)\cdot z\ne\bot\wedge x\cdot (y\cdot z)\ne\bot\wedge(x\cdot y)\cdot z=x\cdot (y\cdot z)))
\\
\label{cmpm5}
&\quad\leftrightarrow((x\ne\bot\wedge y\ne\bot\wedge z\ne\bot)\to(x\cdot y)\cdot z=x\cdot (y\cdot z)),
\\[1.42mm]
\nonumber
\eqref{pm5C}{:}\quad 
&\psitr(\forallp x.\forallp y.~x\cdot y\eqp y\cdot x)
\\
\nonumber
&\quad\leftrightarrow((x\ne\bot\wedge y\ne\bot)\to(x\cdot y\ne\bot\wedge y\cdot x\ne\bot\wedge x\cdot y=y\cdot x))
\\
\label{cmpm6}
&\quad\leftrightarrow((x\ne\bot\wedge y\ne\bot)\to x\cdot y=y\cdot x),
\\[1.42mm]
\nonumber
\eqref{pm6C}{:}\quad 
&\psitr(\forallp x.~1\cdot x\eqp x)
~\equiv~(x\ne\bot\to(1\cdot x\ne\bot\wedge x\ne\bot\wedge 1\cdot x=x))
\\
\label{cmpm7}
&\quad\leftrightarrow (x\ne\bot\to 1\cdot x=x),
\\[1.42mm]
\nonumber
\eqref{pm7C}{:}\quad 
&\psitr(\forallp x.\forallp y.\forallp z.~x \cdot (y+z) \eqp (x \cdot y) + (x \cdot z))
\\
\nonumber
&\quad\leftrightarrow~((x\ne\bot\wedge y\ne\bot\wedge z\ne\bot)\to\\
\nonumber
&\qquad\quad~
( x\cdot(y+z)\ne\bot\wedge(x\cdot y)+(x\cdot z)\ne\bot\wedge x \cdot (y+z) = (x \cdot y) + (x \cdot z)))
\\
\label{cmpm8}
&\quad\leftrightarrow((x\ne\bot\wedge y\ne\bot\wedge z\ne\bot)\to x \cdot (y+z) = (x \cdot y) + (x \cdot z)),
\\[1.42mm]
\nonumber
\eqref{pm8C}{:}\quad 
&\psitr(\forallp x.\forallp y.~y\nep 0\leftimp \frac{x}y\eqp x\cdot \frac 1y)
\equiv(\psitr(\forallp x.\forallp y.y\eqp 0\leftor \frac{x}y\eqp x\cdot \frac 1y))
\\
\nonumber
&\quad\leftrightarrow((x\ne\bot\wedge y\ne \bot)\to( y=0 \vee(y\ne 0\wedge \frac{x}y\ne\bot \wedge x\cdot \frac 1y\ne\bot \wedge \frac{x}y=x\cdot \frac 1y)))
\\
\label{cmpm11}
&\quad\leftrightarrow((y\ne\bot\wedge y\ne 0)\to \frac{x}y=x\cdot \frac 1y),
\\[1.42mm]
\nonumber
\eqref{pm9C}{:}\quad 
&\psitr(\forallp x.~x\nep 0\leftimp \frac xx\eqp 1)
\equiv(\psitr(\forallp x.x\eqp 0\leftor \frac xx\eqp 1))
\\
\nonumber
&\quad\leftrightarrow( x\ne\bot\to(x= 0\vee(x\ne 0\wedge\frac xx\ne\bot\wedge\frac xx= 1)))
\\
\label{cmpm12} 
&\quad\leftrightarrow((x\ne\bot\wedge x\ne 0) \to\frac xx= 1),
\\[1.42mm]
\eqref{pm10C}{:}\quad
&~\psitr(0\nep 1)
\equiv\psifa(0\eqp 1)
\equiv (0=\bot \vee 1=\bot \vee 0\ne 1)
\label{cmpm9}
\leftrightarrow( 0\ne 1).
\end{align}
\hrule
\end{table}

\begin{proof}
\label{page23}
In~Table~\ref{tab:psipm} we display the $\psitr$-transformations of~\eqref{pm1C}--\eqref{pm10C}
of \FTCpmc\ (Table~\ref{tab:FTCpm2b}) and Assertion~\eqref{pmComm+} (Proposition~\ref{prop:pm}), 
while omitting universal quantification in the resulting formulae, which are further simplified
using~\eqref{eq:cm0}--\eqref{eq:cm5}. 

Below we show that the fifteen equational axioms of $\FTCcm$ are consequences of these resulting formulae and  
\eqref{eq:cm0}--\eqref{eq:cm5}. 
Hence, $\psitr(\FTCpmc)\vdash \FTCcm$. 
Since the inequation $0\ne 1$
is also a consequence of one of these resulting formulae, namely~\eqref{cmpm9}, 
$\psitr(\FTCpmc)$ axiomatises a common meadow.

We adopt the convention that $\cdot$ binds stronger than $+$. 
\\[2.04mm]
\underline{Axiom~\eqref{cm1}}. 
By \eqref{eq:cm1}, $(x=\bot \vee y=\bot \vee z=\bot)\to (x+y)+z=x+(y+z)$, hence by \eqref{cmpm1},
$(x+y)+z=x+(y+z)$.
\\[2.04mm]
\underline{Axiom~\eqref{cm2}}. By~\eqref{eq:cm1} and~\eqref{cmpm7},
$(x=\bot \vee y=\bot)\to x+y=y+x$. By 
Proposition~\ref{prop:pm}.\eqref{pmComm+} we may use~\eqref{cmpm2},  
hence $x+y=y+x$.
\\[2.04mm]
\underline{Axiom~\eqref{cm3}}. By \eqref{eq:cm1}, $x=\bot \to x+0=x$, hence by \eqref{cmpm3}, $x+0=x$.
\\[2.04mm]
\underline{Axiom~\eqref{cm4}}. By \eqref{cmpm4}, $x\ne\bot\to x+(-x)=0$. By Prop.\ref{prop:pm}\eqref{eq:0} we 
may use $\psitr(\forallp x.0\cdot x\eqp 0)$, 
and with $0\ne\bot$ \eqref{eq:cm0} this yields
\begin{equation}
\label{0x=0}
x\ne\bot\to 0\cdot x = 0,
\end{equation}  
hence $x\ne\bot\to x+(-x) = 0\cdot x$. 
By \eqref{eq:cm1}, $x=\bot\to x+(-x) = 0\cdot x$. 
Hence $x+(-x) = 0\cdot x$. 
\\[2.04mm]
\underline{Axiom~\eqref{cm5}}. Using~\eqref{cmpm5} and~\eqref{cmpm6},
similar to the case for Axiom~\eqref{cm1}, hence 
$(x\cdot y)\cdot z=x\cdot (y\cdot z)$.
\\[2.04mm]
\underline{Axiom~\eqref{cm6}}. Using~\eqref{cmpm5} and~\eqref{cmpm6},
similar to the case for Axiom~\eqref{cm1}, hence 
$x\cdot y=y\cdot x$.
\\[2.04mm]
\underline{Axiom~\eqref{cm7}}. By~\eqref{eq:cm1}, $x=\bot\to 1\cdot x=\bot$, hence by~\eqref{cmpm7}, $1\cdot x=x$.
\\[2.04mm]
\underline{Axiom~\eqref{cm8}}. By similar reasoning using~\eqref{cmpm8}, $x\cdot (y+z)=(x\cdot y)+(x\cdot z)$.
\\[2.04mm]
\underline{Axiom~\eqref{cm9}}. 
By~\eqref{eq:cm1} and~\eqref{eq:cm5}, $x\ne\bot\leftrightarrow -x\ne \bot$. 
So, by \eqref{cmpm4},
$x\ne\bot\leftrightarrow -x+-(-x)=0$, hence
$x\ne\bot\to x=x+(-x+-(-x))=(x+(-x))+-(-x)=-(-x)$. 
By~\eqref{eq:cm1}, $x=\bot\to x=-(-x)$. 
Hence $-(-x)=x$.
\\[2.04mm]
\underline{Axiom~\eqref{cm10}}. By~\eqref{eq:cm1}, $x+\bot=\bot$.
\\[2.04mm]
\underline{Axiom~\eqref{rep1}}.  By~\eqref{cmpm11}, $(y\ne\bot\wedge y\ne 0)\to \frac xy=x\cdot \frac 1y$
and by~\eqref{eq:cm1} and~\eqref{eq:cm3}, 
$(y=\bot \vee y= 0)\to$ 
$\frac xy=x\cdot \frac 1y$,
hence $\frac xy=x\cdot \frac 1y$.
\\[2.04mm]
\underline{Axiom~\eqref{rep2}}. By~\eqref{eq:cm1} and~\eqref{eq:cm3}, $\frac 1x = \bot\to(x=\bot \vee x=0)$,
hence $(x\ne\bot\wedge x\ne 0)\to$ $\frac 1x\ne\bot$. 
Under this last condition, by~\eqref{0x=0}, $0\cdot\frac 1x=0$ and by~\eqref{rep1}, $0\cdot \frac 1x=\frac 0x$. 
Hence by \eqref{cmpm12}, 
$\frac{x}{x} = 1+0=1+\frac 0x$.
\\[2.04mm]
\underline{Axiom~\eqref{rep3}}. By Proposition~\ref{prop:pm},
\(\psitr(\forallp y.\forallp v.(y \nep 0  \leftand  v\nep 0)\leftimp 
\tfrac xy\cdot \tfrac uv\eqp \tfrac{x\cdot u}{y\cdot v}),
\)
thus 
\begin{equation}
\label{eq:rep3}
(y \ne \bot \wedge v\ne \bot\wedge y\ne 0 \wedge  v\ne 0)\to\tfrac xy\cdot \tfrac uv= \tfrac{x\cdot u}{y\cdot v},
\end{equation}
hence 
$(x \ne \bot \wedge y\ne \bot \wedge x\ne 0 \wedge  y\ne 0)\to\frac 1{x\cdot y} = \frac 1x \cdot \frac 1y$.
By~\eqref{eq:cm1} and~\eqref{eq:cm3}, also
\\
$(x=\bot\vee y=\bot\vee x=0\vee y=0)\to \frac 1{x\cdot y} =\frac 1x \cdot \frac 1y$.
Hence $\frac 1{x\cdot y} = \frac 1x \cdot \frac 1y$.
\\[2.04mm]
\underline{Axiom~\eqref{rep4}}. 
By~\eqref{eq:cm0}, \eqref{eq:cm1} and~\eqref{eq:cm5},  $x\ne\bot\leftrightarrow 1+0\cdot x\ne \bot$.
By~\eqref{cmpm9} and~\eqref{0x=0},  
$1+0\cdot x\ne 0$.\footnote{%
  By~\eqref{cmpm9}
  and~\eqref{0x=0}, 
  $x\ne\bot\to 0\ne 1=1+0\cdot x$ and by \eqref{eq:cm0}--\eqref{eq:cm1},
  $x=\bot\to 0\ne 1+0\cdot x$. Hence, $0\ne 1+0\cdot x$.
  } 
Under these conditions, it follows by~\eqref{cmpm12} that
$\frac{1+0\cdot x}{1+0\cdot x}=1$.
By~ \eqref{eq:cm1} and~\eqref{cmpm7}, $1\cdot x=x$.
By~\eqref{eq:cm1} and~\eqref{eq:cm5}, $x\ne\bot\leftrightarrow x\cdot x\ne \bot$, so
by~\eqref{0x=0}, 
$x\ne\bot\to 0\cdot(x\cdot x)=0$, and hence 
by~\eqref{0x=0},
$x\ne\bot\to 0\cdot x = 0\cdot(x\cdot x)$.
By~\eqref{eq:cm1},
$x=\bot\to 0\cdot x = 0\cdot(x\cdot x)$, so $0\cdot x=0\cdot(x\cdot x)$.
\\[.8mm]
With $1\cdot x=x$, $0\cdot x=0\cdot(x\cdot x)$, and $0\cdot x+0\cdot x=(0+0)\cdot x=0\cdot x$ derive
$(1+0\cdot x)\cdot(1+0\cdot x)=1+0\cdot x$ and therewith
$
(1+0\cdot x)
= \frac{1+0\cdot x}{1+0\cdot x}\cdot\frac{1+0\cdot x}1
\stackrel{\eqref{eq:rep3}}= \frac{(1+0\cdot x)\cdot(1+0\cdot x)}{1+0\cdot x}
= \frac{1+0\cdot x}{(1+0\cdot x)\cdot(1+0\cdot x)}
\stackrel{\eqref{eq:rep3}}= \frac{1+0\cdot x}{1+0\cdot x}\cdot\frac1{1+0\cdot x}
= \frac 1{1+0\cdot x}$.
\\[.8mm]
By~\eqref{eq:cm1}, $x=\bot\to 1+0\cdot x=\frac1{1+0\cdot x}$. Hence 
$1+0\cdot x=\frac1{1+0\cdot x}$.
\\[2.04mm]
\underline{Axiom~\eqref{cm15}}. By~\eqref{eq:cm3}, $\bot=\frac 10$.
\end{proof} 

In the above proof, axiom~\eqref{pm11C} of \FTCpmc, 
i.e., 
\(
\forallp x.\forallp y.\:
(x\nep 0\leftand y\nep 0)
\leftimp x\cdot y \nep 0,
\)
is not used. We will return to this point in Section~\ref{sec:6}, where we address the 
question whether the axioms of $\FTCpm$ are independent. 

\section{Concluding remarks}
\label{sec:6}

Short-circuit logics (SCLs) were introduced in~\cite{BPS13} and are distinguished by 
sequential connectives that prescribe left-to-right (sequential) evaluation of their operands,
in particular, $\fa\leftand x=\fa$, while $x\leftand\fa= \fa$ is not necessarily true (in the case of 
a partial meadow, take $(\frac10 \eqp 0)$ for $x$). 

Depending on the strength of possible atomic side-effects, different \emph{short-circuit} logics were
defined and axiomatised, both for the two-valued and three-valued case, see~\cite{PonseS2018,BPS21}. 
In~\cite{BP23}, three-valued \CLSCL\ is introduced, which differs from Guzmán and Squier's Conditional 
logic~\cite{GS90} only by the use of sequential connectives.
\CLSCL\  with only the constants \tr\ and \fa\ and none for the value \undefi\ (\CLSCLtwo)
is also introduced in~\cite{BP23} and
most closely resembles propositional logic: side-effects are not modelled and full left-sequential
conjuction
$\fulland$, definable by $x\fulland y = (x\leftand y)\leftor (y\leftand x)$, is commutative.
Moreover, adding $x\leftand\fa= \fa$ to \CLSCLtwo\ yields a sequential version of propositional logic
and excludes the use of a third truth value \undefi\ (because with a constant \und\ for 
undefined, it would follow that $\und=\und\leftand\fa=\fa$).

Starting from \CLSCLtwo\ and partial equality 
($\eqp$), we here introduced partial meadows together with axioms and rules split into a 
part for general, partial $\Sigma$-algebras (using rules for weak substitution
from~\cite{BergstraBTW81}) and a part specific to partial meadows (with signature $\Sigma_m^{pd}$). 
We have taken a pragmatic approach and included only axioms and rules that were used in our proofs.
The new quantifiers $\forallp$ and $\existsp$ were introduced for readability and comprehensibility, 
but could have been replaced by their familiar counterparts $\forall $ and $\exists$.

It is an open question whether the axioms of FTCpm ($\FTCpm$ in Table~\ref{tab:FTCpm2}) are independent.
It is certainly  the case that axioms~\eqref{pm1}--\eqref{pm7} are independent (\emph{Mace4})
and imply $x+y\eqp y+x$ (Prop.~\ref{prop:pm}), and it seems that the two axioms~\eqref{pm8} and~\eqref{pm9}
for division, i.e. 
\[y\nep 0\leftimp\frac xy\eqp x\cdot\frac 1y \quad\text{and}\quad x\nep 0\leftimp \frac xx\eqp 1,
\] are both mutually 
independent and also from the first seven. 
Axioms~\eqref{pm10} and \eqref{pm11}, i.e. $0\nep 1$ and
$(x\nep 0\leftand y\nep 0)\leftimp x\cdot y\nep 0$, express
that a partial meadow is an expansion of a field, and their independence  is not clear.

We expect that an extension of \FTCpm\ can be designed, including its proof system, for which a suitable 
completeness theorem can be obtained so that it becomes rewarding to investigate the model theory of 
these axioms, perhaps in a manner comparable to~\cite{DiasDinis2023,DiasDinis2024} where the model theory 
of the axioms for common meadows has been worked out in considerable detail. 

\bigskip

Partial data types for arithmetic arise in different ways, for instance the transreals of Anderson et 
al.~\cite{AndersonVA2007,dosReisGA2016} can be turned
into  partial  transreals where division is more often defined than in a partial meadow 
($\frac{1}{0} = +\infty$, and $\frac{-1}{0} = -\infty $, while $\frac{0}{0}$ is undefined) and where 
addition and multiplication are partial: $\infty + (-\infty) $ and $0 \cdot \infty$ are undefined 
(while in transreals $\infty + (-\infty) = 0 \cdot \infty =\Phi$). 
Just as in a common meadow where $\frac{1}{0} = \bot$ which we replace by undefined to obtain a partial 
meadow, viewing $\Phi$ as a representation of undefinedness leads to partial transreals and to its 
substructure of partial transrationals.

The entropic transreals of~\cite{BergstraT2025} are a modification of transreals which can be turned into partial 
entropic transreals where division and multiplication are total while addition is partial. 
Again the idea is to have $t$ undefined (in the partial entropic transreals) in case $t = \bot$ in entropic 
transreals. In the partial entropic transreals one has $\frac{1}{0} = +\infty$, 
$\frac{0}{0} = 0 \cdot \infty = 0$, while $ \infty + (-\infty) $ is undefined.
 
Wheels (see Carlström's~\cite{Carlstroem2004}) can be made partial also by having $t$ undefined if 
$t = \bot$ in a wheel. In a partial wheel division, addition, and multiplication are each partial. 
We notice that a wheel contains a single unsigned infinite element $\infty$ rather than a pair of 
signed infinite elements $+\infty$ and $-\infty$ (as in transreals and entropic transreals) and 
that in a partial wheel, $\infty + \infty$ is undefined.

\appendix
\section{Prover9 and Mace4 results}
\label{A.0}
The axioms of \CLe\ imply the following consequences, all of which follow easily using 
the theorem prover \emph{Prover9}~\cite{Prover9}:
\begin{description}
 \setlength\itemsep{-.4mm}
\item[\eqref{DNS}]
$\neg\neg \phi=\phi$,
\item[\eqref{(I)}]
$\phi\leftand \tr=\phi$ and $\tr\leftor \phi=\tr$,
\item[\eqref{(II)}]
$(\phi\leftand \psi)\leftand \xi=\phi\leftand(\psi\leftand \xi)$, so the connective $\leftand$ is associative, 
\item[\eqref{(III)}]
$\phi\leftand(\psi\leftand \phi)=\phi\leftand \psi$, so, with $\psi=\tr$, $\leftand$ is idempotent, 
i.e. $\phi\leftand \phi=\phi$,  
\item[\eqref{(IV)}]
$\phi\leftand (\psi\leftor \xi)=(\phi\leftand \psi)\leftor (\phi\leftand \xi)$, so 
$\leftand$ is left-distributive, 
\item[\eqref{(V)}]
$\phi\leftand \neg \phi=\neg \phi\leftand \phi$, 
\item[\eqref{(VI)}]
$(\phi\leftand \psi)\leftimp \xi=\phi\leftimp(\psi\leftimp \xi)$.
\end{description}
 On a Macbook Pro 
with a 2.4 GHz dual-core Intel Core i5 processor and 4 GB RAM, a proof of \eqref{DNS} 
with options \texttt{kbo} and \texttt{pass} requires 0.01~seconds, 
and with \eqref{DNS} added to $\CLe$, a proof of~\eqref{(II)} requires less than 18~seconds. 
With consequence \eqref{(II)} also added to $\CLe$, proofs of the other consequences need less than 1.5~seconds each.

Also Inference rules (i2)--(i5) follow quickly with help of
\emph{Prover9}~\cite{Prover9} (with constants $\mathtt{a}$, $\mathtt b$, $\mathtt c$, 
which then represent arbitrary instantiations).
A proof of (i5) with the default options (\texttt{lpo} and \texttt{unfold}) is the most time-consuming and 
takes less than 1 second.
A trivial proof of inference rule (i4) is as follows: 
if $\vdash \phi=\tr$ and $\vdash \psi=\tr$, then by axiom~\eqref{Tand}, 
$\phi\leftand \psi=\tr\leftand \psi=\psi=\tr$.

With the tool \emph{Mace4}~\cite{Prover9}, a model of \CLe\ 
that refutes $x\leftor y=y\leftor x$ 
is quickly obtained (``\texttt{seconds = 0}''). 
With a constant \und\ for the truth value \undefi, the counterexample in this model
amounts to $\und=\und\leftor\tr\ne\tr\leftor\und=\tr$. 

\end{document}